\documentclass[sigplan,10pt]{acmart}

\renewcommand\footnotetextcopyrightpermission[1]{}
\makeatletter
\renewcommand\@ACM@checkaffil{}
\makeatother

\usepackage[utf8]{inputenc}
\usepackage[T1]{fontenc}
\usepackage{amsmath}
\usepackage{graphicx}
\usepackage{booktabs}
\usepackage{microtype}
\usepackage{hyperref}
\usepackage{algorithm}
\usepackage{xspace}
\usepackage[noend]{algpseudocode}
\usepackage{pbalance}

\usepackage{enumitem}
\usepackage{amsthm}
\usepackage{array}
\usepackage{xcolor}
\usepackage{hyperref}
\usepackage{cleveref}
\usepackage{comment}
\usepackage{makecell}
\usepackage{subcaption}
\usepackage{multirow}

\newtheorem{theorem}{Theorem}
\newtheorem{lemma}{Lemma}

\newenvironment{proofsketch}{\begin{proof}[Proof sketch]}{\end{proof}}

\newcommand{\sectionref}[1]{\S~\ref{sec:#1}}
\newcommand{\appendixref}[1]{\appendixname~\ref{app:#1}}
\newcommand{\figureref}[1]{Figure~\ref{fig:#1}}
\newcommand{\tableref}[1]{Table~\ref{tab:#1}}
\newcommand{\theoremref}[1]{Theorem~\ref{thm:#1}}
\newcommand{\lemmaref}[1]{Lemma~\ref{lem:#1}}

\newcommand{\algorithmref}[1]{Algorithm~\ref{alg:#1}}

\definecolor{dkgreen}{rgb}{0,0.6,0}

\newcounter{evalq}
\newcommand{\evalquestion}[3]{
    \refstepcounter{evalq}\label{eq:#2}
    \textbf{Q\theevalq: #1} #3
}

\newcommand{\stdssd}{\texttt{1-SSD}\xspace}
\newcommand{\highssd}{\texttt{4-SSD}\xspace}

\newcommand{\challenger}{\mathcal{C}}
\newcommand{\adversary}{\mathcal{A}}
\newcommand{\searchsystem}{\mathsf{ANNS}}
\newcommand{\publicparams}{\mathsf{param}}

\newcommand{\negl}{\mathsf{negl}}
\newcommand{\posmap}{\mathsf{PosMap}}
\newcommand{\stash}{\mathsf{Stash}}
\newcommand{\metadata}{\mathsf{Meta}}
\newcommand{\AccessCount}{\mathsf{AccessCount}}
\newcommand{\Bandwidth}{\mathsf{Bandwidth}}
\newcommand{\Enc}{\mathsf{Enc}}
\newcommand{\Dec}{\mathsf{Dec}}

\newcommand{\oramver}{\mathsf{ver}}
\newcommand{\oramptrs}{\mathsf{ptrs}}
\newcommand{\oramvalid}{\mathsf{valid}}
\newcommand{\oramcount}{\mathsf{count}}
\newcommand{\dummycount}{\mathsf{dummy}}
\newcommand{\oramaddrs}{\mathsf{addrs}}

\newcommand{\oramoffset}{\mathsf{offset}}
\newcommand{\oramdata}{\mathsf{data}}
\newcommand{\oramop}{\mathsf{op}}
\newcommand{\oramsk}{\mathsf{sk}}
\newcommand{\oramaad}{\mathsf{aad}}
\newcommand{\pmleaf}{\mathsf{leaf}}
\newcommand{\pmlvl}{\mathsf{lvl}}
\newcommand{\pmslot}{\mathsf{slot}}

\newcommand{\idx}{\mathcal{I}}

\newcommand{\queryvec}{x_q}
\newcommand{\pointvec}{x_p}
\newcommand{\embvec}[1]{x_{#1}}
\newcommand{\query}{q}
\newcommand{\point}{p}
\newcommand{\Setup}{\mathsf{Setup}}
\newcommand{\Search}{\mathsf{Search}}
\newcommand{\Insert}{\mathsf{Insert}}
\newcommand{\Delete}{\mathsf{Delete}}
\newcommand{\GreedySearch}{\mathsf{GreedySearch}}
\newcommand{\RobustPrune}{\mathsf{RobustPrune}}
\newcommand{\traversal}{\mathsf{trv}}
\newcommand{\pruning}{\mathsf{prn}}
\newcommand{\dist}{\mathsf{dist}}

\newcommand{\traversalhints}{\mathcal{H}_\traversal}
\newcommand{\pruninghints}{\mathcal{H}_\pruning}
\newcommand{\prunedlist}{\mathcal{L}_\pruning}
\newcommand{\resultlist}{\mathcal{R}}
\newcommand{\prunedlistsize}{L_\pruning}
\newcommand{\candidatelistsize}{L}
\newcommand{\candidatelist}{\mathcal{L}}
\newcommand{\visitedlist}{\mathcal{V}}
\newcommand{\fullprecision}{\mathcal{F}}
\newcommand{\adjacencylists}{\mathcal{N}}

\newcommand{\toolname}{\textsc{Onyx}\xspace}

\title{\toolname: Cost-Efficient Disk-Oblivious ANN Search}
\author{Deevashwer Rathee}
\authornote{This work was done during an internship at NVIDIA.}
\affiliation{
  \institution{UC Berkeley, NVIDIA}
}

\author{Jean-Luc Watson}
\affiliation{
  \institution{NVIDIA}
}

\author{Zirui Neil Zhao}
\affiliation{
  \institution{UT Austin, NVIDIA}
}

\author{G. Edward Suh}
\affiliation{
  \institution{NVIDIA}
}

\author{Raluca Ada Popa}
\affiliation{
  \institution{UC Berkeley}
}

\begin{document}

\begin{abstract}
Approximate nearest neighbor (ANN) search in AI systems increasingly handles sensitive data on third-party infrastructure.
Trusted execution environments (TEEs) offer protection, but cost-efficient deployments must rely on external SSDs, which leaks user queries through disk access patterns to the host.
Oblivious RAM (ORAM) can hide these access patterns but at a high cost; when paired with existing disk-based ANN search techniques, it makes poor use of SSD resources, yielding high latency and poor cost-efficiency.

The core challenge for efficient oblivious ANN search over SSDs is balancing both bandwidth and access count.
The state-of-the-art ORAM-ANN design minimizes access count at the ANN level and bandwidth at the ORAM level, each trading-off the other, leaving the combined system with both resources overutilized.
We propose inverting this design, minimizing bandwidth consumption in the ANN layer and access count in the ORAM layer, since each component is better suited for its new role: ANN's inherent approximation allows for more bandwidth efficiency, while ORAM has no fundamental lower bounds on access count (as opposed to bandwidth).
To this end, we propose a cost-efficient approach, \toolname, with two new co-designed components:
\toolname-ANNS introduces a compact intermediate representation that proactively prunes the majority of bandwidth-intensive accesses without hurting recall, and \toolname-ORAM proposes a locality-aware shallow tree design that reduces access count while remaining compatible with bandwidth-efficient ORAM techniques.
Compared to the state-of-the-art oblivious ANN search system, \toolname achieves $1.7$--$9.9\times$ lower cost and $2.3$--$12.3\times$ lower latency.
\end{abstract}

\maketitle
\pagestyle{plain}

\section{Introduction} \label{sec:intro}

Approximate nearest neighbor (ANN) search powers many AI systems, including retrieval-augmented generation (RAG) \cite{rag, rag-text-gen-survey}, search~\cite{colbert, passage-reranking-bert}, and recommendation~\cite{collaborative-filtering, graph-cnn-recommender}.
As these systems increasingly run on third-party infrastructure and operate over sensitive data, privacy has become a critical requirement.
For example, systems such as ChatGPT serve hundreds of millions of users and perform RAG queries over their highly sensitive chat messages~\cite{chatgpt-memory}.

A natural first step to protect the search index and query contents is to encrypt them.
However, it is well-known from over a decade of leakage-abuse attacks~\cite{islam2012access, cash2015leakage, zhang2016leakage, grubbs2017leakage, blackstone-attack, oya2022ihop, nie2024jigsaw, jia2025fit} that encryption alone is not enough: a service provider observing the sequence of accesses during query processing can infer information about the query and the dataset, and sometimes, recover the dataset entirely~\cite{lambregts2022val}.
Cryptographic approaches such as oblivious RAM (ORAM)~\cite{goldreich-ostrovsky} hide this leakage, but remain too expensive for ANN search. Given a 10M-vector index at $90\%$ recall, the state-of-the-art solution Compass~\cite{compass} serves $100$ queries/\$ and incurs over 1~s of latency, far from the under-20~ms latencies expected in production deployments~\cite{azure-cosmos-db-vector}.

To avoid the high overhead of purely cryptographic approaches, a practical alternative is to rely on hardware assumptions and run ANN search inside a trusted execution environment (TEE)~\cite{amd-sev-snp, intel-tdx}, where the index is stored in the protected TEE memory.
However, keeping the entire index in enclave memory quickly becomes expensive at scale~(\sectionref{ann-usecases}).
In practice, cost-efficient ANN systems store the index on SSDs~\cite{diskann, spann, starling, fusionanns, fresh-diskann, bang}, but this storage typically lies outside the TEE trust boundary, introducing access pattern leakage through disk I/O (\sectionref{threat-model}).

\textbf{ORAM-in-TEE for disk-oblivious accesses.}
In principle, the solution is simple: run Compass~\cite{compass}, the state-of-the-art ORAM-ANN system, inside a TEE to hide disk access patterns.
However, this approach is unsuitable for the disk setting as it places \emph{heavy demands on SSD I/O}, ultimately requiring more compute and SSD resources and driving up cost.
For instance, sustaining $\approx 50$~QPS on a 20M $3$~KB vector index at 90\% recall requires at least 4~vCPUs and 8~SSD units, and still incurs $\approx 160$~ms latency~(\sectionref{eval-compass-comparison}).

\begin{figure}[t]
    \centering
    \includegraphics[width=0.95\columnwidth, trim={1.05cm, 48.5cm, 52cm, 0cm}, clip]{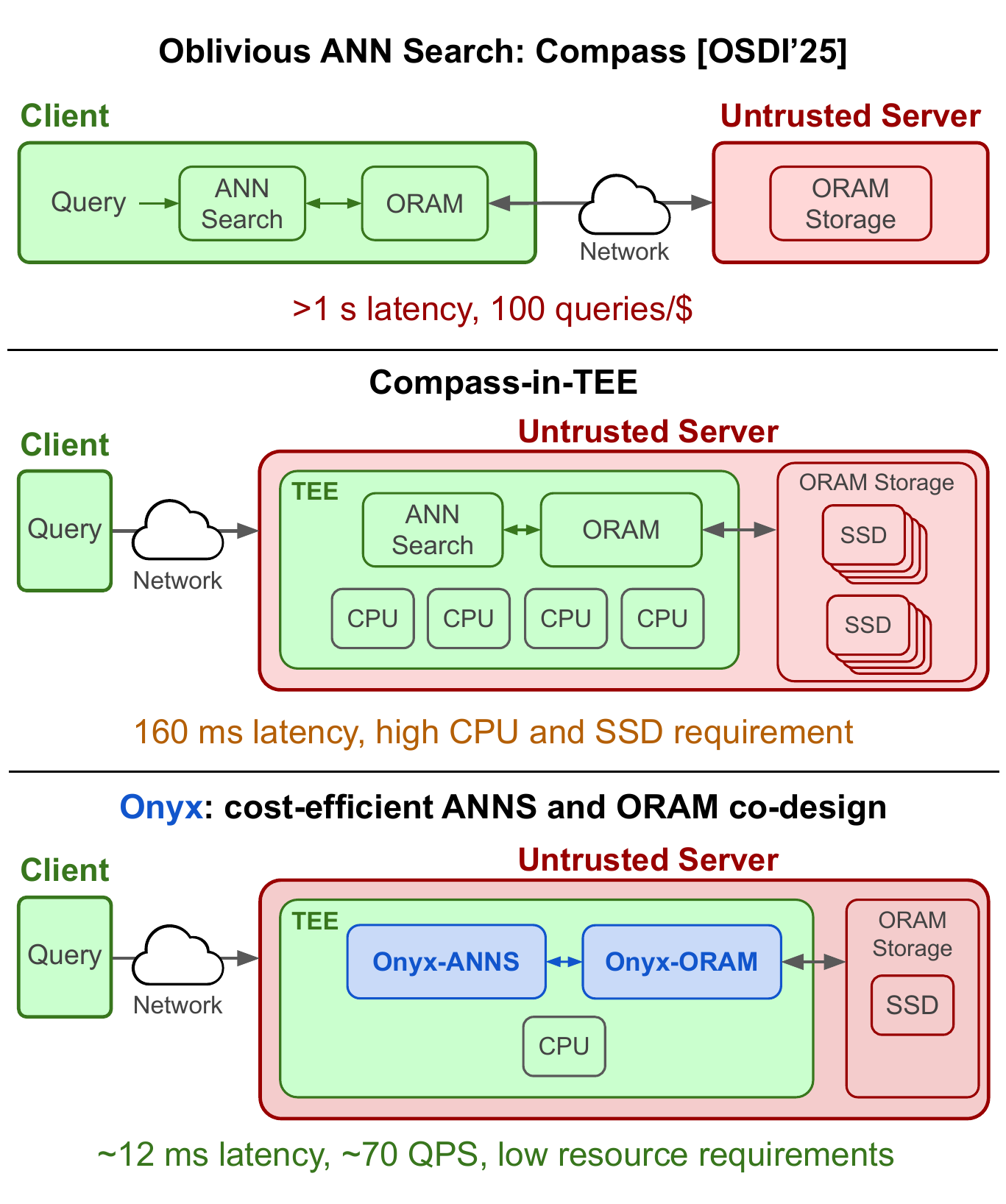}
    \caption{Prior designs for oblivious ANNS and \toolname. Compass~\cite{compass} suffers from poor performance due to high network overhead, and mirroring its design into a TEE yields a system that needs significant CPU and SSD resources to sustain reasonable performance. \toolname proposes co-designed primitives for ANNS and ORAM that jointly improve resource utilization, yielding high performance and low cost.}
    \label{fig:intro-fig}
\end{figure}

To address this, we propose \toolname, a fundamentally different ANN-ORAM co-design \emph{for the disk setting} that significantly reduces this overhead through efficient utilization of SSD resources (see \figureref{intro-fig} for an overview).
\toolname achieves up to $10\times$ better cost-efficiency than the Compass-in-TEE approach; with a single vCPU and a single SSD unit on the same index, it sustains $\approx 70$~QPS at $12$~ms latency.

\textbf{Inverting the ORAM-ANN co-design.}
The core challenge for efficient oblivious ANN search over SSDs is balancing \emph{SSD bandwidth} and \emph{SSD access count}, since both directly determine performance and resource cost.
Prior ORAM-ANN co-design~\cite{compass}, optimizing for oblivious search over a network, combines an access-count-efficient ANN (reducing network round-trips) with a bandwidth-efficient ORAM (reducing communication bandwidth), but each primitive aggressively trades off the other resource. The ANN layer fetches larger blocks to reduce accesses, wasting bandwidth due to the inherently low spatial locality of ANN search~\cite{starling, diskann++}, while the ORAM layer makes many fine-grained accesses to reduce data movement but is ultimately limited by fundamental lower bounds to ORAM bandwidth~\cite{oram-lower-bound, oram-lower-bound-all-params}.

In this work, we invert this design, arguing that both ORAM and ANN are better-suited for their new role.
ANN search is naturally approximate, so it can maintain high recall with smaller, lower-precision representations, supporting aggressive bandwidth reduction.
ORAM has no fundamental lower bounds on access count, allowing it to significantly reduce the number of I/O operations required per access.

\textbf{New primitives for our co-design.}
Existing primitives fall short of realizing the optimization objectives of our design.
ANN search algorithms for disk-hosted indices are designed assuming bandwidth is plentiful, which is true for plain SSDs but not for ORAM over SSDs.
On the ORAM side, access-efficient designs such as EnigMap~\cite{enigmap} introduce high bandwidth amplification, negating the ANN-side savings.
To address both, we propose two new co-designed primitives.
\textbf{\toolname-ANNS}~(\sectionref{onyx-anns}) retains the low memory footprint and accesses required for disk-hosted indices, while introducing a compact intermediate representation that proactively prunes the majority of bandwidth-intensive full-precision accesses without hurting recall. This reduces ANN bandwidth by up to $5\times$ while increasing access count by at most $15\%$ compared to the state-of-the-art~\cite{diskann} disk-based ANN search~(\sectionref{onyx-anns-overview}).
\textbf{\toolname-ORAM}~(\sectionref{onyx-oram}) proposes a locality-aware shallow tree design that significantly reduces access count while remaining compatible with bandwidth-efficient ORAM techniques. Compared to EnigMap-style access-optimized PathORAM~\cite{enigmap}, our design achieves $2\times$ lower access count and $4.6\times$ lower bandwidth simultaneously~(\sectionref{onyx-oram-overview})
Together, the two primitives achieve a much better balance of SSD bandwidth and access count than any prior combination.

\textbf{Evaluation summary.}
We implement and evaluate \toolname in \sectionref{evaluation}.
Overall, \toolname is consistently the most cost-efficient and lowest-latency approach for TEE-based oblivious ANN search at scale. Compared to Compass-in-TEE, it achieves $1.7$--$9.9\times$ lower cost and $2.3$--$12.3\times$ lower latency across all datasets~(\sectionref{eval-end-to-end}).
We also evaluate \toolname against other combinations of state-of-the-art ORAM and ANN
techniques. Compared to the best of these baselines, \toolname achieves $2.6$--$4.2\times$
lower cost and $2.5$--$4.5\times$ lower latency.
Given a resource partition with just 1~vCPU, 4~GB RAM, and 187.5~GB SSD, \toolname can host a 64~GB index and sustain 70~QPS at 12~ms latency with 90\% top-$10$ recall, serving over 8M queries/dollar.

\section{Background}\label{sec:background}
We provide background on our setting below. Details on the specific ORAM and ANN primitives we build on are deferred to the sections where we discuss our designs (\sectionref{onyx-oram} and \sectionref{onyx-anns}).

\subsection{Motivation: Private Disk-hosted ANN Search} \label{sec:ann-usecases}
We discuss two deployment scenarios where cost and scalability necessitate private disk-hosted ANN indices.

\textbf{Scenario I: Enterprise Knowledge Bases.} \label{sec:motivation-enterprise}
This scenario targets enterprise applications (e.g., RAG for internal semantic search) hosted on third-party clouds~\cite{azure-cosmos-db-vector, zilliz-milvus, pinecone, weaviate}. To protect proprietary data while leveraging cloud scalability, organizations use Confidential VMs~\cite{azure-confidential-computing, gcp-confidential-computing}.
The primary motivation for disk-based indexing is \emph{cost-efficiency}.
Applications that do not require very high throughput (e.g. 100 QPS) often cannot justify the expense of keeping the entire index in memory.
For example, a single 20M-vector partition with 3~KB vectors~\cite{wiki-cohere} requires $64$~GB DRAM to host in memory, but only $\approx 2$~GB memory with a disk-hosted index~(\sectionref{eval-datasets}).

\textbf{Scenario II: Personalized Data Stores.} \label{sec:motivation-personal}
This scenario targets multi-tenant services where each user maintains a private, isolated index.
A key use case is private memory for the next generation of personalized AI agents (e.g., Microsoft Recall~\cite{microsoft-recall}, Gemini Personal Intelligence~\cite{gemini-personal-intelligence}), where highly sensitive personal information such as browsing history, messages, and documents is consolidated into per-user indices.
The combination of privacy, low latency, and scalability makes disk-hosted indices inside a TEE a natural fit.
Existing privacy-preserving solutions such as Compass~\cite{compass} rely on client-side processing, incurring high-latency network round trips.
To achieve the interactive latencies required by real-time agents, confidential computing within clouds offers a natural solution.
In this architecture, keeping the full indices of millions of users constantly in cloud DRAM is economically infeasible.
A disk-hosted index only needs a small in-memory cache per user to serve low-latency requests: for example, a 1M-vector index~\cite{sift} requires only 8~MB of in-memory hints instead of the 640~MB index~(\sectionref{eval-datasets}).

\subsection{Disk Access Pattern Leakage}\label{sec:attack-mechanism}
We focus on the direct leakage from disk I/O during ANN search on disk-hosted indices.
The SSD typically lies outside the TEE trust boundary, and the host software stack can directly observe the sequence of disk accesses issued during query execution.
Even if the index and query contents are encrypted, the access pattern alone can reveal both.

Over a decade of research on leakage-abuse attacks~\cite{islam2012access, cash2015leakage, zhang2016leakage, grubbs2017leakage, blackstone-attack, damie2021refined, oya2022ihop, lambregts2022val, nie2024jigsaw, jia2025fit} has shown that such access patterns enable highly effective query recovery, and from that, reconstruction of the dataset itself.
Even in a passive setting, where the attacker only observes access patterns, the Refined Score attack~\cite{damie2021refined} recovers $\approx$85\% of queries starting from just $\approx$10 known queries, and optimization-based attacks such as IHOP~\cite{oya2022ihop} and Jigsaw~\cite{nie2024jigsaw} achieve near-perfect recovery using only an auxiliary dataset (e.g., a public corpus approximating the target data).
Active attacks amplify the threat: because indices ingest data from many sources (e.g., emails, shared documents), an adversary can inject carefully crafted entries that encode identifiers into the access pattern so that each future query produces a uniquely decodable response.
File-injection attacks~\cite{zhang2016leakage, blackstone-attack} show that tracking $n$ identifiers requires only $O(\log n)$ injected entries.

While majority of these attacks were developed for keyword search, recent work \textsc{FiT}~\cite{jia2025fit} focuses on access pattern leakage during semantic search, showing that page-level accesses (exactly the granularity leaked in our setting) suffice to infer the specific vectors fetched within each page for graph-based ANN indices like HNSW.
Since ANN search is typically followed by a retrieval of the corresponding data for the $K$ nearest vectors, the adversary can learn both the query's neighborhood in embedding space and the returned documents.
Encrypting the index is insufficient; a secure disk-based ANN search must make query-oblivious accesses to the untrusted storage, which is the goal of our work.

\section{System Overview} \label{sec:system-overview}
\toolname is a privacy-preserving semantic search system that allows querying a database while keeping the queries and the dataset private.
To protect privacy, it executes within a trusted execution environment (TEE), keeping the core logic and data in TEE-protected memory, and for cost efficiency~(\sectionref{ann-usecases}), it stores the index outside the TEE trust boundary on untrusted external storage (e.g., NVMe SSDs).

\paragraph{Design overview.}
\toolname consists of two co-designed components running inside the TEE trust boundary: a \emph{bandwidth-efficient} ANN search engine (\toolname-ANNS) and an \emph{access-count-efficient} ORAM client (\toolname-ORAM).
\toolname-ANNS introduces compact \emph{pruning hints} and a three-step granular refinement strategy~(\sectionref{onyx-anns-overview}) that ensures the majority of external storage accesses fetch only small blocks (adjacency lists and pruning hints, 256--512 bytes each), with only a few accesses fetching the larger full-precision embedding vectors.
These accesses are mediated through two separate \toolname-ORAM instances: one for the small traversal blocks and one for the larger refinement blocks.
\toolname-ORAM uses a locality-aware shallow tree design~(\sectionref{onyx-oram-overview}) that balances both access count and bandwidth to achieve high throughput at the traversal block sizes.
The two components reinforce each other: since ORAM performance is sensitive to block size, \toolname-ANNS's smaller blocks directly reduce the \toolname-ORAM's per-access overhead,
and \toolname-ORAM achieves its greatest throughput advantage precisely at the 256--512 byte block sizes that dominate \toolname-ANNS accesses~(\figureref{access-vs-bandwidth-oram}).

\paragraph{Operations.}
Within the broader semantic search pipeline, which consists of both an embedding model to map data items to vectors and a vector database to do a similarity search over them, \toolname specifically focuses on the vector database component.
It supports the following operations over the embedding vectors:
\begin{itemize}[leftmargin=*]
    \item $\Setup(\idx)$: Given an ANN index $\idx$, initialize both ORAMs over the index and load in-memory index state.
    \item $\resultlist \gets \Search(\queryvec, k)$: Given the query embedding vector $\queryvec$, returns the keys $\resultlist$ of the $k$ nearest neighbors.
    \item $\Insert(\point, \pointvec)$: Add vector $\pointvec$ to index with key $\point$.
    \item $\Delete(\point)$: Remove the entry corresponding to key $\point$.
\end{itemize}
To avoid expensive per-block initialization, one can use prior work on bulk-loading the initial ORAM contents such as BULKOR~\cite{bulkor} and EnigMap~\cite{enigmap}.
In the rest of the paper, we primarily focus on search performance.
At the same time, the same design choices we make to improve search can also improve insertion and deletion performance (\appendixref{dynamic-index}).

\paragraph{Data layout.}
\toolname splits state across TEE-protected memory and two on-disk ORAM-protected layouts.
\begin{itemize}[leftmargin=*]
    \item \textbf{Inside TEE (protected memory).} (i) ANN metadata and the start node for graph traversal, (ii) following prior work~\cite{diskann, compass}, highly compressed in-memory traversal hints $\traversalhints$ for approximate distance computation during search, and (iii) ORAM client state for both ORAM instances (e.g., position maps, stashes, keys).
    \item \textbf{Outside TEE (external storage, untrusted).} The full index is stored within two separate ORAM instances:
    (i) a \emph{traversal ORAM} holding per-node blocks consisting of adjacency lists $\adjacencylists$ concatenated with compact pruning hints $\pruninghints$, and
    (ii) a \emph{refinement ORAM} holding full-precision embedding vectors $\fullprecision$.
\end{itemize}

\paragraph{Paper organization.}
We present \toolname-ORAM in \sectionref{onyx-oram} and \toolname-ANNS in \sectionref{onyx-anns}, starting with ORAM because ANNS builds on top of it.
Both sections present: background on their base primitive~(\sectionref{background-oram}, \sectionref{background-ann}), design overview~(\sectionref{onyx-oram-overview}, \sectionref{onyx-anns-overview}), construction~(\sectionref{onyx-oram-construction}, \sectionref{onyx-anns-construction}), and I/O analysis~(\sectionref{oram-analysis}, \sectionref{onyx-anns-analysis}).
We evaluate \toolname in \sectionref{evaluation} and discuss related work in \sectionref{related-work}.

\subsection{Threat Model}\label{sec:threat-model}

\paragraph{Trusted computing base.}
Following a common TEE threat model~\cite{intel-tdx, amd-sev-snp}, we assume that the CPU package is trusted and functions correctly.
The application code running inside the TEE is also assumed to be correctly implemented and free of software vulnerabilities.
All data residing within the TEE's private memory is protected by hardware-level encryption and isolation; DRAM access patterns are out of scope (see scope and exclusions below).

\textbf{Adversary capabilities.}
Like prior TEE systems~\cite{weave, varys, aex-notify, heisenberg-defense, opaque, haven, vc3, scone, ryoan, enclavedb, safebricks, occlum}, we assume that the adversary does not compromise the TEE but can observe and manipulate its interactions with untrusted interfaces like storage.
Our security guarantees hold even if the adversary knows the stored vector index and the distribution of user queries.
Since the TEE communicates with an \emph{untrusted} disk via shared memory (e.g., bounce buffers~\cite{bounce-buffer}) visible to the hypervisor, the adversary can \emph{directly observe} the disk access pattern generated during query execution.
It can also manipulate the responses returned from disk unless prevented by integrity checks.
Its goal is to infer information about the query/dataset from the observed disk access pattern.

We focus on disk access patterns because disk-based storage is essential for cost efficiency and scalability~(\sectionref{ann-usecases}), and even a perfect TEE cannot prevent this leakage: an untrusted disk inherently lies outside the trust boundary, and a cloud provider can trivially inspect the I/O trace.
As discussed in \sectionref{attack-mechanism}, the observed access pattern can reveal the exact top-$K$ results for a given query, making disk access obliviousness a first-order privacy requirement for deploying cost-efficient ANN search on third-party cloud providers.

\textbf{Scope and exclusions.}
A compromised TEE is out of scope for our work.
This includes cache-based side channels~\cite{yarom2014flushreload, cachezoom}, branch prediction attacks~\cite{lee2017branchshadowing, branchscope}, controlled-channel attacks via page faults and interrupts~\cite{xu2015controlled, vanbulck2017sgxstep, tdxdown}, transient execution attacks~\cite{meltdown, spectre}, software-based fault injection~\cite{cachewarp, rowhammer, plundervolt}, timing and performance counter leakage~\cite{t-time, countersev}, physical attacks~\cite{lee2020membuster, teefail, wiretap, batteringram, badram}, remote-attestation attacks~\cite{teefail, wiretap, batteringram}, and denial-of-service attacks.
These are the subject of an orthogonal line of work on improving enclave implementations and mitigations.
For example, controlled-channel attacks can be mitigated by hardware-assisted interrupt protection~\cite{aex-notify, heisenberg-defense, tlblur}, and cache-based leakage by partitioning and core isolation~\cite{composable-cachelets, varys}.
Unlike these enclave implementation issues, disk access patterns arise outside the enclave and cannot be mitigated by improved design---we therefore focus exclusively on disk access pattern leakage.
We also assume that an orthogonal rollback protection mechanism (e.g., monotonic counters or distributed freshness protocols~\cite{matetic2017rote, strackx2016ariadne, nimble}) prevents the adversary from replaying stale state from the persistent storage.
The query type is public and not hidden by our system.

\subsection{Formalizing Security}
We formalize the privacy guarantees of \toolname\ through an indistinguishability-based security game.
We give a high-level description here and defer the full specification to \figureref{disk-security-game} (\appendixref{security-proof}).
Let $\searchsystem$ denote a disk-resident ANN search system that executes within a TEE and accesses an untrusted storage interface through ORAM, with public parameters $\publicparams$.
At a high level, the challenger $\challenger$ executes the protocol and maintains the private state of the system, which remains hidden from the adversary since TEE side channels are out of scope, while the adversary $\adversary$ plays the role of the malicious host and untrusted storage.
Whenever the protocol issues an external storage access, the request is sent to $\adversary$, which observes the access trace and can return an arbitrary response.
The adversary adaptively chooses two equal-size indices and a sequence of paired operations (constrained to have matching types), and wins if it can distinguish which of the two worlds the challenger is executing.
If the challenger aborts at any point, the adversary loses.

\begin{definition}
    We say that a disk-resident ANN search system $\searchsystem$ provides \emph{disk-access privacy} if no non-uniform probabilistic polynomial-time adversary $\adversary$ can win the security game in \figureref{disk-security-game} (\appendixref{security-proof}) with probability non-negligibly higher than random guessing.
\end{definition}

\begin{theorem}
\label{thm:disk-privacy}
    \toolname is disk-access private when instantiated with a secure authenticated encryption scheme~\cite{authenticated-encryption}.
\end{theorem}
\begin{proofsketch}
The ORAM and ANN search parameters are configured statically and are public.
The ANN operation type is public, and the number and granularity of (logical) disk accesses made by each ANN operation depends only on these public parameters~(\sectionref{onyx-anns-construction}).
Each (logical) disk access is mediated through ORAM, which guarantees data-oblivious accesses with confidentiality and integrity given authenticated encryption~\cite{authenticated-encryption}.
\toolname-ORAM builds on RingORAM~\cite{ringoram}, and the changes we introduce~(\sectionref{onyx-oram-construction}) preserve all the protocol invariants and the security guarantees.
Thus, either the adversary deviates from the protocol and the challenger aborts, or the two traces seen by the adversary are indistinguishable.
We provide the full proof in \appendixref{security-proof}.
\end{proofsketch}

\section{\toolname-ORAM} \label{sec:onyx-oram}

\subsection{Background: RingORAM}\label{sec:background-oram}
\toolname-ORAM builds on RingORAM~\cite{ringoram}, a tree-based ORAM that hides access patterns to untrusted storage.
We use the non-recursive variant (the TEE holds the position map) without the XOR trick (the SSD has no compute capability).

For $N$ blocks, RingORAM organizes storage as a binary tree over buckets with depth $L = \Theta(\log{N})$.
Each block is mapped by a position map $\posmap$ to a random leaf and resides either on that root-to-leaf path or in a client-local stash.
Every $A$ accesses, a background \emph{eviction} process flushes stash blocks along a deterministic path that cycles through the tree, while ensuring each block remains on its assigned path.
Each bucket holds at most $Z$ real blocks, and is padded to $Z+S$ blocks with dummies. A bucket can be accessed at most $S$ times before it must be fetched, reshuffled, and rewritten (\emph{early reshuffle});
$S$ is configured so that eviction reshuffles the bucket with high probability before this limit is reached.

\textbf{Access.}
To read a block, the client first checks the stash, then fetches a root-to-leaf path (the block's assigned path, or a random path if the block was in the stash).
For each bucket on the path, it reads a small metadata header to determine which block to read (the real block if present, otherwise an unread dummy).
The accessed block is then added to the stash, remapped to a fresh random path, and returned.

\textbf{Eviction.}
Every $A$ accesses, eviction reads $Z$ random blocks per bucket on the eviction path (since buckets have at most $Z$ real blocks), assigns stash blocks to eviction buckets, adds dummies, shuffles and encrypts the buckets, and then writes back the full $Z+S$ buckets to the tree.

\textbf{I/O multipliers.}
The dominant cost for ORAM over SSD is SSD I/O, which comes from two sources: the \emph{access count multiplier} (I/O requests per logical access) and the \emph{bandwidth multiplier} (bytes transferred per logical access).
Ignoring the rare early reshuffles, RingORAM's I/O multipliers are:
$\AccessCount \approx \left(2 + \frac{Z + 1}{A}\right)\cdot L$ and $\Bandwidth \approx \left(1 + \frac{2Z+S}{A}\right)\cdot L$.
The leading $2$ in the access count comes from the per-level metadata fetch. The $(Z+1)/A$ term reflects eviction access count amortized over $A$ accesses.

\subsection{Design Overview} \label{sec:onyx-oram-overview}

\paragraph{Goal: access-efficient ORAM}
Recall that our co-design requires an access-count-efficient ORAM.
This, however, should not sacrifice bandwidth efficiency.
For the block sizes typical of our bandwidth-efficient ANN search (256--512 Bytes), both multipliers contribute significantly~(\sectionref{oram-analysis}), and our ORAM must balance both.

\textbf{Tension between access count and bandwidth.}
To hide access patterns over $N$ blocks, state-of-the-art tree-based ORAMs~\cite{ringoram, pathoram} make random accesses to $O(\log N)$ tree levels per logical access.
This is the core bottleneck for access count, since the remaining ORAM operations (e.g., evictions) can be made deterministic and sequential (\sectionref{onyx-oram-construction}, step 2).
To reduce this, EnigMap~\cite{enigmap} proposed packing $\ell$ tree levels into a single SSD page for PathORAM~\cite{pathoram}, thereby reducing tree depth and access count by $\ell\times$.
This, however, leads to high bandwidth: packing increases bandwidth by $\frac{2^{\ell} - 1}{\ell}\times$, and the design is incompatible with the large buckets used by bandwidth-efficient ORAMs like RingORAM.

\textbf{Our approach: locality-aware shallow trees.}
In \toolname-ORAM, we propose reducing tree depth directly by using a \emph{$d$-ary tree}, which cuts depth and access count by $\ell = \log_2 d\times$ while remaining compatible with large buckets.
$d$-ary trees were introduced over a decade ago~\cite{gentry-oram} but fell out of use: although they reduce tree depth, they increase the frequency of bandwidth-heavy evictions by $(d{-}1)\times$, growing bandwidth by roughly $d/\log d\times$ and incurring a similar kind of trade-off as EnigMap's level packing.

For oblivious accesses to SSDs, however, we observe that this design point becomes attractive for large-bucket ORAMs like RingORAM.
RingORAM uses large buckets with many dummy slots to allow fetching individual blocks rather than entire buckets, saving significant bandwidth.
The key observation is that in a $d$-ary tree, although evictions are $(d{-}1)\times$ more frequent, the number of dummy slots per bucket decreases proportionally (since fewer accesses occur between evictions), and the two effects balance each other out.
This is especially prominent with our locality-aware optimizations, which require fetching even more dummies to keep the evictions sequential.
For $d = 8$, this effect is so prominent that the $d$-ary tree ends up not increasing eviction bandwidth at all while reducing access count by $3\times$~(\sectionref{onyx-oram-construction}, step 3).
Additionally, the $d$-ary tree with fewer dummies reduces storage amplification by $4\times$, further improving cost-efficiency.
\tableref{oram-overhead-breakdown} summarizes the I/O multipliers of \toolname-ORAM vs prior work: \toolname-ORAM improves PathORAM $(\ell = 3)$ access count by $2\times$ while reducing its bandwidth by $4.6\times$.

\subsection{Construction} \label{sec:onyx-oram-construction}

\begin{table}[t]
\centering
\small
\begin{tabular*}{\columnwidth}{@{\extracolsep{\fill}}lcc}
\toprule
 & \textbf{Access count} & \textbf{Bandwidth} \\
\midrule
\vspace{5pt}
RingORAM & $\approx \frac{5}{2} \cdot \log(N)$ & $\approx 3 \cdot \log(N)$ \\
\vspace{5pt}
PathORAM $(\ell=2)$ & $\approx \log(N)$  & $\approx 12 \cdot \log(N)$ \\
\vspace{5pt}
PathORAM $(\ell=3)$ & $\approx \frac{2}{3} \cdot \log(N)$  & $\approx 18.5 \cdot \log(N)$ \\
\toolname-ORAM & $\approx \frac{1}{3} \cdot \log(N)$ & $\approx 4 \cdot \log(N)$ \\
\bottomrule
\end{tabular*}
\caption{I/O multipliers for \toolname-ORAM vs prior work. PathORAM $\ell$: levels packed per page. RingORAM bandwidth shown is the best case, achieved in the large-bucket limit.}
\label{tab:oram-overhead-breakdown}
\end{table}

Recall the access count multiplier for RingORAM from \sectionref{background-oram} is $\AccessCount \approx \left(2 + \frac{Z + 1}{A}\right)\cdot L$,
where $Z$ is the bucket capacity, $A$ is the eviction frequency, and $L = \log_2 N$ is the tree depth.
We reduce this multiplier to $\approx L/3$ in \toolname-ORAM, and we do so in two stages.
First, we apply \emph{locality-aware optimizations} that eliminate the multiplicative factor before $L$, reducing the per-level cost from $2 + (Z+1)/A$ down to $1$, i.e., a single block read per level.
Second, we reduce the tree depth itself from $L = \log_2 N$ to $\approx L/3$ by switching to a wider, shallower tree.
Crucially, the depth reduction is only effective once the locality-aware optimizations are in place and it is particular well-suited for bandwidth-efficient ORAMs with large buckets.
We describe these changes as a sequence of steps, each stacking on the previous one, and then discuss how the three steps reinforce each other to yield a much better trade-off than any one in isolation.
\algorithmref{oram-protocol} in \appendixref{oram-protocol} gives the full protocol.

\textbf{Step 1: Store bucket metadata locally and increase bucket size.}
In RingORAM, each access fetches small per-bucket metadata at every level of the path to determine which block within the bucket to read~(\sectionref{background-oram}).
This adds one I/O per level, accounting for the leading $2$ (instead of $1$) in the access count expression, and also limits bucket size to small values (e.g., $Z = 32$, $128$) for smaller blocks because metadata storage grows with $Z$.
\toolname-ORAM stores this metadata locally to avoid a per-level metadata fetch.
This requires at least 10 bytes of memory per block, which can more than double the ORAM memory footprint~(\appendixref{oram-memory-storage}); as we will see, our later design choices substantially minimize this overhead.
With the metadata cost removed, we are free to use much larger buckets, which in turn allows a proportionally larger eviction period $A$.
The larger $A$ better amortizes eviction cost, and together these changes reduce the access count multiplier to $\AccessCount \approx (1 + (Z+1)/A) \cdot L$.

\textbf{Step 2: Fetch entire buckets during evictions.}
RingORAM's eviction reads only $Z$ blocks from each bucket on the eviction path (as opposed to $Z+S$), as there could be at most $Z$ real blocks in a bucket.
Since $S$ is typically $\approx 2Z$, this saves significant bandwidth, but it means the eviction must issue $Z$ separate reads per bucket, which increases access count.
\toolname-ORAM instead reads each bucket in full ($Z + S$ slots) as a single sequential I/O.
This trades additional bandwidth for far fewer accesses, reducing the access count multiplier to $\AccessCount \approx (1 + 2/A) \cdot L \approx 1 \cdot L$ since $A$ can be set large enough to make the eviction term negligible.

\textbf{Step 3: Use a shallower tree.}
\toolname-ORAM replaces the binary tree with a $d$-ary tree ($d > 2$), reducing depth from $\log_2 N$ to $\log_d N$ and the access count by a factor of $\log_2 d$.
The trade-off is more frequent evictions: we generalize RingORAM to $d$-ary trees and found that evictions must be $(d{-}1)\times$ more frequent to keep the stash bounded~(\sectionref{oram-analysis}).
Since the tree has $\log_2 d$ fewer levels, one would na\"ively expect eviction bandwidth to increase by $(d{-}1)/\log_2 d$, or about $2.3\times$ for $d = 8$.
For RingORAM's large buckets, however, reducing $A$ has the benefit of proportionally reducing $S$ and the dummy padding, which added significant bandwidth pressure.
This offsets the higher eviction frequency, reducing the net bandwidth increase to roughly $\frac{d+1}{3\log_2 d}\times$, which for $d = 8$ is $\approx 1$.
As a result, this step yields a $3\times$ access count reduction with essentially no bandwidth penalty: $\AccessCount \approx \frac{L}{3}$.

The three steps above are not independent; they reinforce each other and \emph{together yield a much better trade-off than any one in isolation}:

\textbf{Step~1 enables Step~3.} A $d$-ary tree requires large $Z$ to properly amortize the more frequent evictions; for instance, with $d = 8$, the ratio $Z/A$ is $6.4$ for $Z = 32$ but drops to $4.13$ for $Z = 256$. Large buckets are efficient because Step~1 eliminates the per-access metadata overhead that previously scaled with $Z$.

\textbf{Step~2 enables Step~3.} Without full-bucket reads, evictions would still issue $Z$ individual reads per bucket, and the $Z/A$ access count term (which increased $\frac{d-1}{\log{d}}\times$) would make the access count worse not better.

\textbf{Step~3 benefits Step~1.} The reduced padding and fewer buckets in a $d$-ary tree shrink the otherwise high per-block metadata memory footprint: the overhead of storing bucket metadata locally drops from 10 bytes/block ($d = 2$) to just 2.5 bytes/block ($d = 8$). For the same reasons, ORAM storage amplification is also reduced by $4\times$~(\appendixref{oram-memory-storage}).

\textbf{Step~3 benefits Step~2.} Step~2 increased eviction bandwidth from $(2Z{+}S)/A$ to $(2Z{+}2S)/A$ per level, making bandwidth more sensitive to the dummy padding $S$. In a $d$-ary tree, $S$ shrinks proportionally with $A$, which reduces this pressure and actually ends up decreasing eviction bandwidth by $\approx 1.1\times$ when we go from $d = 2$ to $d = 8$.

\begin{figure}[t]
    \centering
    \includegraphics[width=0.85\columnwidth]{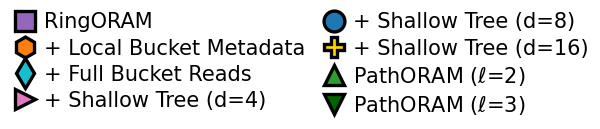}\\
    \includegraphics[width=0.85\columnwidth]{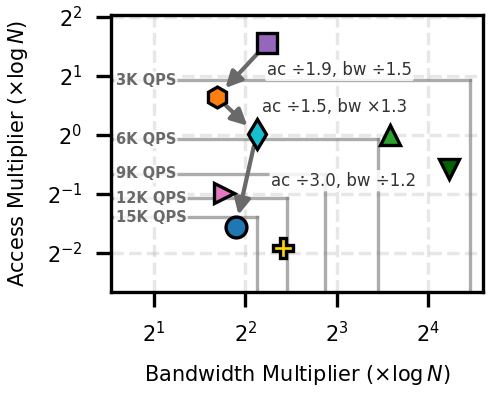}
    \caption{I/O overhead design space for ORAM schemes at $B = 512$ Bytes. Gray iso-throughput contours show the maximum I/O multipliers that our evaluation SSD~(\sectionref{eval-setup}) can theoretically sustain at a given throughput. Arrows trace \toolname-ORAM's step-by-step improvements from RingORAM, annotated with the factor change in each multiplier.}
    \label{fig:access-vs-bandwidth-oram}
\end{figure}

\subsection{Analysis} \label{sec:oram-analysis}

\paragraph{Performance analysis.}
\figureref{access-vs-bandwidth-oram} shows the maximum analytical throughput each ORAM design can sustain on our evaluation SSD with $512$-byte blocks, computed from the I/O multipliers and the SSD's bandwidth and IOPS ceilings (see~\sectionref{eval-oram} for throughput results with actual implementation).
RingORAM and PathORAM are each limited by one of the two multipliers (RingORAM by access count, PathORAM by bandwidth) and both lie beyond the $6$K QPS zone.
\toolname-ORAM ($d = 8$) reduces both multipliers simultaneously, reaching the $15$K QPS zone.
This plot also shows that maintaining bandwidth is important: $d = 16$ achieves a lower access count than $d = 8$ but at higher bandwidth, and lands in the $12$--$15$K QPS zone.
This highlights the importance of balancing both multipliers, and $d = 8$ strikes a favorable trade-off for our setting.

\textbf{Eviction analysis.}
\toolname-ORAM generalizes RingORAM's binary tree to a $d$-ary tree while preserving the reverse-lexicographic eviction order and all other protocol invariants.
We state the main result here and defer the full proof, which closely follows the structure of RingORAM's stash analysis~\cite{ringoram}, to \appendixref{eviction-proof}.

\begin{theorem}[Stash overflow in $d$-ary RingORAM]\label{thm:dary-stash}
Consider a $d$-ary \toolname-ORAM with $N$ blocks, bucket capacity $Z$, eviction period $A$, and tree depth $L = \lceil \log_d \frac{2N}{A(d-1)} \rceil$.
Let $a = A(d{-}1)/2$.
If $q = Z \ln(Z/a) + a - Z - 1 - \ln d > 0$, then $\Pr[\mathit{stash} > R] \;\leq\; \frac{(a/Z)^{R}}{1 - e^{-q}}$.
\end{theorem}
\noindent
Setting $d = 2$ recovers the original RingORAM result~\cite{ringoram}.
The key consequence is that the maximum eviction period scales as $A \leq 2Z/(d{-}1)$, compared to $A \leq 2Z$ for binary RingORAM: evictions must be $(d{-}1)\times$ more frequent.

\section{\toolname-ANNS} \label{sec:onyx-anns}

\subsection{Background: DiskANN}\label{sec:background-ann}
\toolname-ANNS builds on DiskANN~\cite{diskann}, a graph-based algorithm that is the state-of-the-art for disk-based ANN search, and is widely deployed in industry~\cite{azure-cosmos-db-vector, zilliz-milvus, pinecone, weaviate, diskann-overview}.

\textbf{Index layout.}
The dataset is organized as a proximity graph on external storage.
For each vector $\point$ with embedding $\pointvec$, the index stores its neighbor list $\adjacencylists(\point)$ and full-precision coordinates $\fullprecision(\point)$ packed together in a contiguous block on disk, so that both can be fetched efficiently with one I/O.
Per-vector quantized hints $\traversalhints(\point)$~\cite{pq-quantized-hints} are stored in memory to approximate distances during traversal, which would otherwise require fetching full-precision vectors for each visited node.

\textbf{Greedy search.}
A greedy beam search navigates the graph to find the top-$K$ vectors closest to a query $\queryvec$.
Beam width $W$ controls how many unvisited candidates are explored concurrently: for each candidate $\point$, the system fetches the combined block $\adjacencylists(\point) \| \fullprecision(\point)$ from disk, scores newly discovered neighbors using the in-memory hints $\traversalhints$, and adds them to the beam.
Since $\fullprecision$ is prefetched alongside $\adjacencylists$, the system maintains an exact top-$K$ ranking throughout the search using the full-precision coordinates.
The search continues until $\candidatelistsize$ nodes have been fetched;
increasing $\candidatelistsize$ improves recall at the cost of more disk accesses.

\subsection{Design Overview} \label{sec:onyx-anns-overview}

\begin{table}[t]
    \centering
    \footnotesize
    \setlength{\tabcolsep}{4pt}
    \begin{tabular*}{\columnwidth}{@{\extracolsep{\fill}}ll cccc@{}}
    \toprule
    & & \textbf{SIFT} & \textbf{MARCO} & \textbf{DEEP} & \textbf{WIKI} \\
    \midrule
    \multirow{2}{*}{\shortstack[l]{\textbf{Bandwidth}\\(reduction)}}
      & Decoupling     & $1.1\times$ & $1.0\times$ & $1.3\times$ & $2.2\times$ \\
      & \toolname-ANNS & $2.7\times$ & $5.0\times$ & $2.0\times$ & $3.8\times$ \\
    \midrule
    \multirow{2}{*}{\shortstack[l]{\textbf{Access Count}\\(increase)}}
      & Decoupling     & $+91\%$ & $+99\%$ & $+72\%$ & $+43\%$ \\
      & \toolname-ANNS & $+8\%$  & $+5\%$  & $+9\%$  & $+15\%$ \\
    \bottomrule
    \end{tabular*}
    \caption{Bandwidth reduction and access count increase of naive decoupling and \toolname-ANNS over DiskANN at top-$10$ $90\%$ recall on our evaluation datasets~(\sectionref{evaluation}).
    }
    \label{tab:onyx_ann_vs_diskann_k10_r95}
\end{table}

\paragraph{Goal: bandwidth-efficient ANN search.}
Recall that our co-design requires a bandwidth-efficient disk-based ANN search.
Existing designs~\cite{diskann, starling, spann} target plain SSDs, where random reads up to 4~KB cost nearly the same. ANN accesses are typically smaller, making access count the core bottleneck and fetching extra bytes effectively free.
With ORAM, however, bandwidth is amplified by $O(\log N)\times$ (\tableref{oram-overhead-breakdown}), and every byte the ANN layer transfers directly affects performance, calling for a new disk-based ANN design that minimizes bandwidth.

\textbf{Decoupling traversal and refinement.}
We start by undoing an access-efficiency optimization that disk-based ANN designs make: fusing traversal and refinement operations.
To achieve high recall, graph traversal must touch many nodes, and this coupling would force our ANN search to fetch large, full precision vectors (384--3072 bytes), consuming excessive bandwidth.
With decoupling, we can follow the standard approach from clustering-based ANNs~\cite{faiss}: use low-precision vectors as hints to prune the traversal candidate list, and fetch full-precision vectors only for the pruned subset.

Naively pruning using the in-memory traversal hints $\traversalhints$, however, is not effective~(\appendixref{ann-ablation}): these hints are often too coarse to filter candidates effectively and require re-ranking nearly the entire candidate list, roughly doubling the access count while barely reducing bandwidth (\tableref{onyx_ann_vs_diskann_k10_r95}).
Increasing traversal hint size to improve pruning precision is also not a viable solution, as it increases the DRAM footprint and thus the deployment cost~(\sectionref{ann-usecases}).

\textbf{Decoupling traversal and pruning hints.}
The core issue with naive decoupling is that it leverages in-memory hints for two tasks with very different precision requirements.
Traversal can tolerate coarse hints because it only needs directional accuracy to move toward the query's neighborhood. Graph connectivity corrects suboptimal local choices through convergent paths.
Pruning, in contrast, is a sensitive filtering step where a mistakenly excluded true neighbor is permanently lost, requiring higher precision.
To address this mismatch, we introduce \emph{compact pruning hints} $\pruninghints$, an intermediate representation that is substantially more precise than $\traversalhints$ but an order of magnitude smaller than $\fullprecision$.
Crucially, these pruning hints are not stored in memory; instead, they are fetched from disk alongside the neighbor list $\adjacencylists$ whenever a node is visited during traversal, keeping the DRAM footprint unchanged.

\textbf{Three-step granular refinement.}
This yields a \emph{three-step granular refinement} strategy (traverse $\to$ prune $\to$ refine) where traversal fetches much smaller blocks ($\adjacencylists \| \pruninghints$ rather than $\adjacencylists \| \fullprecision$) and the majority of expensive full-precision fetches are avoided, all without increasing the memory footprint.
\tableref{onyx_ann_vs_diskann_k10_r95} shows that for the same recall, our approach reduces bandwidth by $2$--$5\times$ across datasets compared to DiskANN, while increasing access count by only $5$--$15\%$.

\subsection{Construction} \label{sec:onyx-anns-construction}

\begin{algorithm}[htb]
\caption{\toolname-ANNS $\GreedySearch$ Algorithm}
\label{alg:onyx-anns-algo}
\begin{algorithmic}[1]
\Statex \hspace{-1.2em}\textbf{Input:} Query vector $\queryvec$
\Statex \hspace{-1.2em}\textbf{Parameters:} Result size $K$, candidate list size $\candidatelistsize$, and pruned list size $\prunedlistsize$ s.t. $K \leq \prunedlistsize \leq \candidatelistsize$
\Statex \hspace{-1.2em}\textbf{Index Layout:} Start node $s$ and traversal hints $\traversalhints$ in memory, per-vector concatenation of pruning hints $\pruninghints$ and neighbor indices $\adjacencylists$ on external storage, and full precision coordinates $\fullprecision$ on external storage
\Statex \hspace{-1.2em}\textbf{Output:} Indices of the $K$-approximate nearest neighbors
\Statex
\State $\candidatelist \gets \{s\}, \visitedlist \gets \emptyset$ \Comment{Candidate and visited lists}
\Statex \textbf{// Phase 1: Traversal}
\While{$\candidatelist \setminus \visitedlist \neq \emptyset$}
    \State let $p^\ast \gets \arg\min_{p \in \candidatelist \setminus \visitedlist}~ \dist(\queryvec, \traversalhints(p))$
    \State fetch $\adjacencylists(p^\ast) \| \pruninghints(p^\ast)$ from \textbf{external storage}
    \State update $\candidatelist \gets \candidatelist \cup \adjacencylists(p^\ast)$ and $\visitedlist \gets \visitedlist \cup \{p^\ast\}$
    \State sort $\candidatelist$ by $\dist(\queryvec, \traversalhints(\cdot))$ and set $\candidatelist \gets \candidatelist[:\candidatelistsize]$
\EndWhile
\Statex \textbf{// Phase 2: Pruning}
\State sort $\candidatelist$ by $\dist(\queryvec, \pruninghints(\cdot))$ using pruning hints
\State set $\prunedlist \gets \candidatelist[:\prunedlistsize]$ \Comment{Pruned candidate list}
\Statex \textbf{// Phase 3: Refinement}
\For{each candidate $p \in \prunedlist$}
    \State fetch $\pointvec \gets \fullprecision(\point)$ from \textbf{external storage}
\EndFor
\State sort $\prunedlist$ by $\dist(\queryvec, \fullprecision(\cdot))$ using full precision vectors
\State set $\resultlist \gets \prunedlist[:K]$ and \textbf{return} $\resultlist$
\end{algorithmic}
\end{algorithm}

\algorithmref{onyx-anns-algo} gives the search procedure for \toolname-ANNS, simplified to beam width $W = 1$ (i.e., one candidate explored per step) for clarity.
In addition to beam width $W$, the algorithm takes three parameters: the result size $K$, the candidate list size $\candidatelistsize$, and the pruned list size $\prunedlistsize$, where $K \leq \prunedlistsize \leq \candidatelistsize$.
To ensure obliviousness, both $\candidatelistsize$ and $\prunedlistsize$ are fixed parameters independent of the query.

\textbf{Index layout.}
\toolname-ANNS reorganizes DiskANN's on-disk layout: $\adjacencylists(\point) \| \pruninghints(\point)$ are packed together (fetched during traversal), while $\fullprecision(\point)$ is stored separately (fetched only during refinement).
Pruning hints $\pruninghints$ are constructed using product quantization~\cite{pq-quantized-hints} at a higher fidelity than $\traversalhints$; we choose $|\pruninghints|$ per dataset to maximize bandwidth reduction.
Only the traversal hints $\traversalhints$ are stored in memory.

\textbf{Memory footprint.}
Since traversal and refinement blocks are now stored separately on disk, they require two independent ORAM clients, which nearly doubles the ORAM-side memory footprint.
Nevertheless, the ORAM client state is typically a small fraction of total memory (the bulk is $\traversalhints$), and we show in \sectionref{eval-compass-comparison} that this leads to only a $10$--$30\%$ increase.

We walk through the three phases below.

\textbf{Phase 1: Traversal (lines 2--6).}
A greedy beam search explores the graph starting from node $s$.
At each step, the top $W$ unvisited candidates (rather than one, as in the simplified \algorithmref{onyx-anns-algo}) are selected by approximate distance $\dist(\queryvec, \traversalhints(\cdot))$.
For each, the system fetches $\adjacencylists(\point) \| \pruninghints(\point)$ from disk in a single I/O.
Newly discovered neighbors are scored using $\traversalhints$ and the beam is updated.
Traversal continues until exactly $\candidatelistsize$ nodes have been visited.
This phase makes $\candidatelistsize$ accesses, each of size $|\adjacencylists| + |\pruninghints|$ bytes.

\textbf{Phase 2: Pruning (lines 7--8).}
The $\candidatelistsize$ candidates are re-ranked using the prefetched pruning hints $\pruninghints$, and only the top $\prunedlistsize$ are retained.
No disk accesses are needed since all pruning hints were prefetched during traversal.

\textbf{Phase 3: Refinement (lines 9--11).}
For each of the $\prunedlistsize$ surviving candidates, $\fullprecision(\point)$ is fetched from disk, candidates are re-ranked by exact distance, and the top-$K$ are returned.
This phase makes exactly $\prunedlistsize$ accesses, each of size $|\fullprecision|$.

\subsection{Analysis} \label{sec:onyx-anns-analysis}

The total access count is $\candidatelistsize + \prunedlistsize$ (vs.\ DiskANN's $\candidatelistsize$).
The per-access bandwidth ratio is: $\frac{\text{BW}_{\text{DiskANN}}}{\text{BW}_{\text{\toolname-ANNS}}} = \frac{|\adjacencylists| + |\fullprecision|}{(|\adjacencylists| + |\pruninghints|) + \frac{\prunedlistsize}{\candidatelistsize} \cdot |\fullprecision|}$
Two terms in the denominator determine the improvement: $(|\adjacencylists| + |\pruninghints|)$, the smaller traversal block, and $\frac{\prunedlistsize}{\candidatelistsize} \cdot |\fullprecision|$, the amortized refinement cost.
These depend on two dataset characteristics.
First, when $|\fullprecision|$ is large, the block size ratio $(|\adjacencylists| + |\fullprecision|) / (|\adjacencylists| + |\pruninghints|)$ is large, since the numerator is dominated by $|\fullprecision|$ while the denominator remains small.
Second, when $|\traversalhints|$ is highly compressed to minimize memory footprint (e.g., for personal databases), $\candidatelistsize$ must be large to achieve high recall with coarse guidance, and pruning hints can filter aggressively, yielding a small $\prunedlistsize / \candidatelistsize$.

Our datasets illustrate both effects nicely (\tableref{onyx_ann_vs_diskann_k10_r95}).
WIKI and MS-MARCO have large embedding vectors ($|\fullprecision| = 3072$~B), and MS-MARCO and SIFT are personal databases that have highly compressed hints $|\fullprecision| / |\traversalhints| \geq 64$.
MS-MARCO benefits from both factors and achieves the largest improvement ($5\times$), followed by WIKI ($3.8\times$) with its large vectors, then SIFT ($2.7\times$) with its better pruning ratio, and finally DEEP, where neither factor applies but we still see $1.8\times$ improvement.

\section{Evaluation}\label{sec:evaluation}

We answer the following questions in this section:
\begin{itemize}[leftmargin=*]
    \item \evalquestion{\toolname-ORAM (\sectionref{eval-oram}).}{oram}{How does \toolname-ORAM perform against state-of-the-art tree-based ORAMs like RingORAM and locality-optimized PathORAM on SSDs?}
    \item \evalquestion{End-to-End Oblivious ANN Search (\sectionref{eval-end-to-end}).}{end-to-end}{How does \toolname compare against state-of-the-art oblivious ANN search and other ORAM-ANN combinations?}
    \item \evalquestion{\toolname Co-design (\sectionref{eval-codesign}).}{codesign}{What is the combined benefit of the \toolname's co-designed primitives, compared to their individual performance?}
    \item \evalquestion{Cost Efficiency (\sectionref{eval-costefficiency}).}{cost-efficiency}{What is the cost benefit (throughput per $\$$) of \toolname compared to state-of-the-art oblivious ANN search systems, across a wide range of hardware configurations?}
\end{itemize}

\subsection{Experimental Setup} \label{sec:eval-setup}
\paragraph{Datasets.} \label{sec:eval-datasets}
We evaluate \toolname on both deployment scenarios from \sectionref{ann-usecases}, adopting two datasets from Compass -- SIFT~\cite{sift} (1M vectors) and MS-MARCO~\cite{msmarco} (8.8M vectors) -- for personal databases and WIKI~\cite{wiki-cohere} (20M vectors) and DEEP~\cite{deep} (60M vectors) for enterprize databases where large indices are sharded into disk-served partitions. Following partitioned disk-based ANN systems~\cite{starling}, we use a DRAM budget of approximately 4~GB per partition, with 2~GB allocated for hints.
Together, these four datasets cover a wide range of index sizes (0.6--64~GB), vector dimensionalities (384~B--3~KB), and hint compression ratios ($12\times$--$96\times$), as summarized in \tableref{dataset-stats}.
We summarize the index construction hyperparameters used in \appendixref{index-hyperparams}.

\begin{table}[t]
\centering
\small
\begin{tabular*}{\columnwidth}{@{\extracolsep{\fill}}lcccc@{}}
\toprule
\textbf{Dataset} & \textbf{\#Vectors} & \textbf{Index Size} & \textbf{Vector Size} & \textbf{Hint Size} \\
\midrule
SIFT & 1M & 0.6 GB & 512 B & 8 B \\
MARCO & 8.8M & 29 GB & 3072 B & 32 B \\
WIKI & 20M & 64 GB & 3072 B & 96 B \\
DEEP & 60M & 31 GB & 384 B & 32 B \\
\bottomrule
\end{tabular*}
\caption{Datasets used in our evaluation. Index size corresponds to the plain graph index (full precision vectors + adjacency lists). Hint size refers to the in-memory hints, significantly compressed to minimize DRAM footprint.}
\label{tab:dataset-stats}
\end{table}

\paragraph{Hardware Configuration.} We run our experiments on a machine with an AMD EPYC 9554 Processor, with AMD SEV-SNP~\cite{amd-sev-snp} enabled, and Micron 7450 MAX SSDs~\cite{micron-7450-max} attached through \texttt{vfio-pci}.
We partition our evaluation setup into \emph{resource units} (RUs) using \texttt{cgroups}, and all experiments are performed on a single RU unless otherwise mentioned.
Since \toolname and each baseline have different SSD requirements (the latter require substantially more SSD resources), we analyze performance on two RU configurations:
a cheaper configuration -- \stdssd\ -- with 1~vCPU, 4~GB of RAM, and 1~SSD unit, providing 85K IOPS, 330~MB/s read bandwidth, 175~MB/s write bandwidth, and 187.5 GB of storage, and a more expensive configuration -- \highssd\ -- with the same vCPU and DRAM, but 4$\times$ the SSD resources. Due to space constraints, we summarize results for the \highssd\ configuration throughout this section, but defer plots to \appendixref{additional-eval}.
Unless otherwise mentioned, we target $K=10$ and 90\% recall, which is standard for benchmarking ANN search~\cite{bigann-benchmarking, bigann-benchmarking-23, ann-benchmarking}.
We estimate dollar cost using GCP's N2D-standard pricing with AMD SEV-SNP enabled and 3-year resource-based committed use discounts (as of March 2026)~\cite{gcp-pricing}.
Under this pricing, vCPUs are \$12.2/month, memory is \$1.2/GB/month, and each local SSD unit is \$6.75/month.

\paragraph{ORAM Parameters.}
For large-bucket ORAMs (RingORAM and \toolname-ORAM), we set the early reshuffle rate to $0.1\%$.
For RingORAM, we vary the bucket capacity $Z$ between 32 and 128, after which point, the large $Z$ capacity increases the overhead for per-bucket metadata fetches such that overall performance is degraded. For \toolname-ORAM, which is optimized for larger bucket sizes, we vary $Z$ between 128 and 1024.
For a given capacity $N$, we pick the $Z$ value that minimizes memory and storage overhead.
Following EnigMap~\cite{enigmap}, we set $Z=4$ for PathORAM and for each block size, vary the level-packing parameter $\ell$ based on how many tree levels fit within a 4~KB SSD page and select the best-performing configuration.

\paragraph{ANN Search Parameters.}
We vary recall by varying the candidate list size $\candidatelistsize$ during the search until the target recall is reached.
For approaches that decouple traversal and refinement, we fix $\candidatelistsize$ to match the coupled baseline and vary $\prunedlistsize$ until we achieve the same recall.
We vary the beam width over $\{4, 8, 16\}$ for DiskANN-based indices (including \toolname-ANN), and the speculation set size and direction filter size $\in$ $[2, 8]$ and $[4, 24]$, respectively, for Compass indices. We target at most $9$ round trips (to match Compass), but if the recall target cannot be met within this budget, we increment the round-trip count until the target is reached.
For all approaches, we vary the number of parallel requests (queue depth) over $\{1, 2, \ldots, 8\}$ to measure the throughput-vs-latency trade-off.
For throughput-vs-recall plots, we limit the total number of visited nodes to $\leq 1500$.

\begin{figure}[!t]
    \centering
    \includegraphics[width=1.02\linewidth]{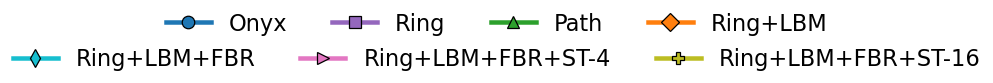}\\[0.2em]
    \includegraphics[width=0.48\linewidth]{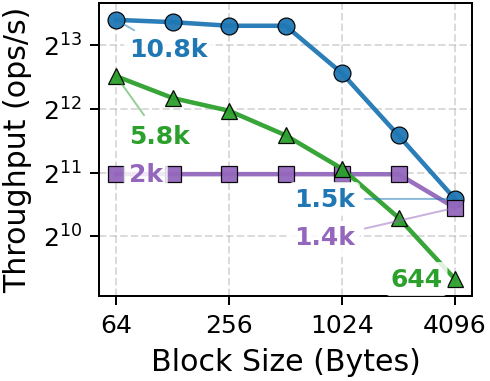} \hfill
    \includegraphics[width=0.48\linewidth]{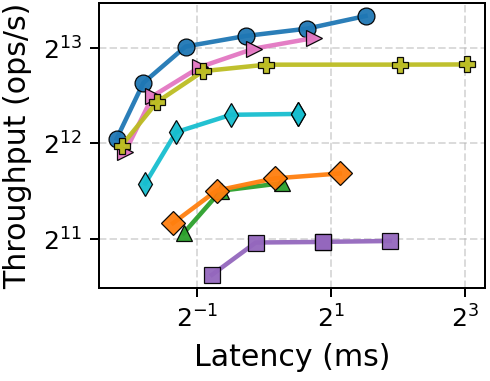}\\[0.2em]
    \makebox[0.48\linewidth]{\small (a) Throughput vs.\ block size} \hfill
    \makebox[0.48\linewidth]{\small (b) Throughput vs.\ latency ($512$~B)}
    \caption{\toolname-ORAM vs prior work for $20$M blocks (\stdssd). LBM (local bucket metadata), FBR (full bucket reads), ST (shallow tree) refer to steps in \sectionref{onyx-oram-construction}; \toolname: Ring+LBM+FBR+ST-8; Ring: RingORAM; Path: locality-optimized PathORAM.}
    \label{fig:oram-eval}
\end{figure}

\subsection{\toolname-ORAM} \label{sec:eval-oram}
We evaluate \toolname-ORAM in isolation against RingORAM~\cite{ringoram} and locality-optimized PathORAM~\cite{pathoram,enigmap} on $20$M block storage, matching the WIKI-20M dataset.
\figureref{oram-eval} provides a detailed comparison between \toolname and each baseline across two performance metrics.

\textbf{Throughput vs.\ block size} (\figureref{oram-eval}(a)).
\toolname-ORAM maintains substantially higher throughput than both baseline ORAM protocols at small (256--512~B) block sizes which are used in ANN search, achieving over 10K logical accesses per second in that range. In contrast, RingORAM is access-count-bound: even at small block sizes, its throughput plateaus around 2K QPS because the high access count per logical operation saturates the SSD's IOPS.
On the other hand, PathORAM is bandwidth-bound: its throughput degrades with increasing block size, falling below RingORAM at 1~KB blocks.
At very large block sizes, \toolname-ORAM also becomes bandwidth-bound and matches the performance of RingORAM because they have similar bandwidth overheads~(\tableref{oram-overhead-breakdown}).

\textbf{Throughput vs.\ latency} (\figureref{oram-eval}(b)).
We evaluate the throughput-latency tradeoff for \toolname and baselines with the addition of several intermediate design points activating the techniques from \sectionref{onyx-oram-construction}. At 512~B, the maximum traversal block size \toolname-ANN uses, \toolname-ORAM achieves up to $5.1\times$ higher throughput and $2.7\times$ lower latency than RingORAM, and up to $3.4\times$ higher throughput and $2.0\times$ lower latency than PathORAM. In sequence, enabling local bucket metadata (LBM), full bucket reads (FBR), and shallow trees with various $d$-arities yield better performance, following the analytical results in \figureref{access-vs-bandwidth-oram} (\sectionref{oram-analysis}). \toolname-ORAM uses each of these optimizations and a $d=8$-ary tree, which both analytically and experimentally yields the best tradeoff between bandwidth and access count overheads.

\begin{figure*}[!t]
    \centering
    \includegraphics[width=0.70\textwidth]{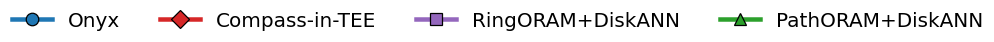}\\[0.2em]
    \rotatebox[origin=c]{90}{\small Throughput (QPS)}\,
    \parbox[c]{0.96\textwidth}{\centering
        \includegraphics[width=0.23\textwidth]{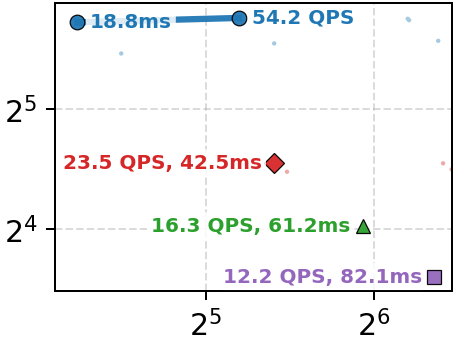} \hfill
        \includegraphics[width=0.23\textwidth]{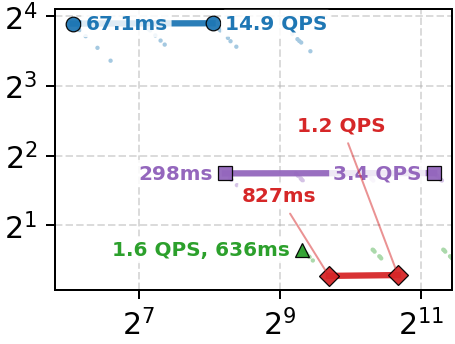} \hfill
        \includegraphics[width=0.23\textwidth]{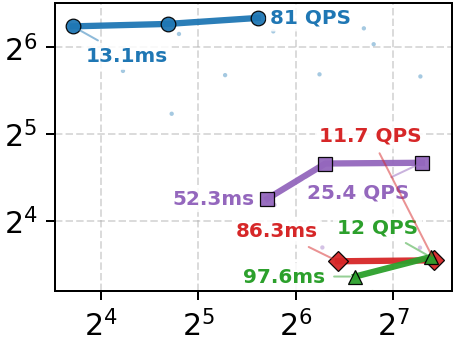} \hfill
        \includegraphics[width=0.23\textwidth]{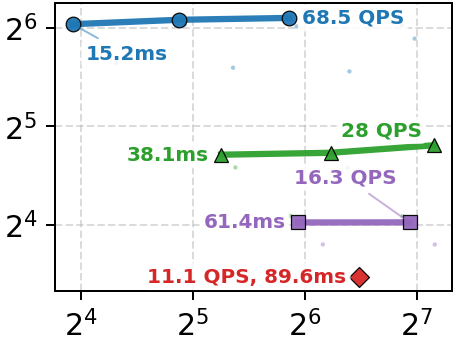}
    }\\[0.1em]
    \hspace*{0.04\textwidth}
    \parbox[t]{0.96\textwidth}{\centering\small Latency (ms)}\\[0.3em]
    \rotatebox[origin=c]{90}{\small Throughput (QPS)}\,
    \parbox[c]{0.96\textwidth}{\centering
        \includegraphics[width=0.23\textwidth]{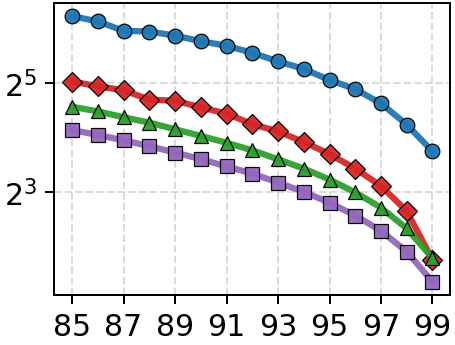} \hfill
        \includegraphics[width=0.23\textwidth]{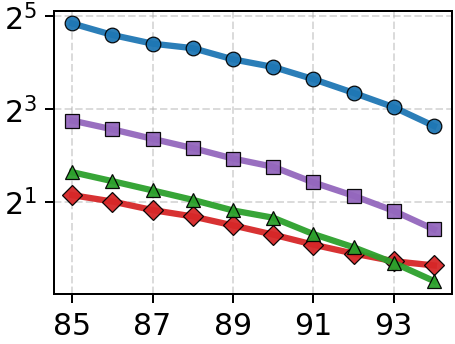} \hfill
        \includegraphics[width=0.23\textwidth]{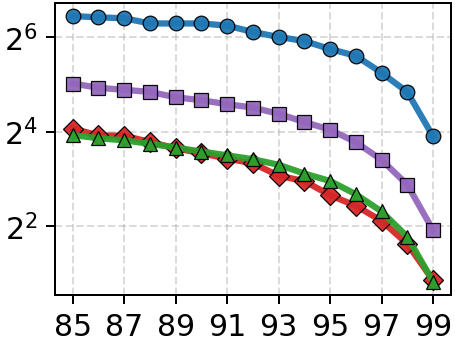} \hfill
        \includegraphics[width=0.23\textwidth]{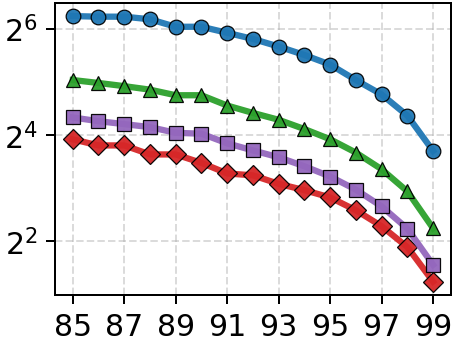}
    }\\[0.1em]
    \hspace*{0.04\textwidth}
    \parbox[t]{0.96\textwidth}{\centering\small Recall}\\[0.2em]
    \hspace*{0.04\textwidth}
    \parbox[t]{0.96\textwidth}{
        \makebox[0.23\textwidth]{\small (a) SIFT} \hfill
        \makebox[0.23\textwidth]{\small (b) MARCO} \hfill
        \makebox[0.23\textwidth]{\small (c) WIKI} \hfill
        \makebox[0.23\textwidth]{\small (d) DEEP}
    }
    \caption{(Top) Pareto frontier of throughput vs.\ latency of \toolname and baselines (\stdssd; top-left is better). Dots below the curves represent sub-optimal parameterizations; single markers are the configuration for a scheme that maximizes both latency and throughput. (Bottom) Throughput vs.\ recall of \toolname and baselines (\stdssd).}
    \label{fig:throughput-vs-latency}
    \label{fig:throughput-vs-recall}
\end{figure*}

\subsection{End-to-End Oblivious ANN Search} \label{sec:eval-end-to-end} \label{sec:eval-compass-comparison}

We compare \toolname against Compass-in-TEE---the state-of-the-art oblivious ANN search system Compass~\cite{compass} deployed inside a TEE to hide disk accesses.
For each system, we issue queries from the dataset's query set with $\{1, 2, 4, 8\}$ parallel requests and measure throughput and latency at each concurrency level.
Compass was designed for oblivious access over a network and the Compass-ANN design was not specifically optimized for disk; to provide stronger baselines, we consider two additional combinations with DiskANN:
RingORAM + DiskANN, which replaces Compass-ANN with DiskANN~\cite{diskann}, and PathORAM + DiskANN, which further replaces RingORAM with locality-optimized PathORAM~\cite{pathoram,enigmap}.
For a fair comparison on SIFT and MS-MARCO, we match Compass's ANN hint sizes.

\textbf{Throughput vs.\ latency.} In \figureref{throughput-vs-latency}, we sweep over number of parallel requests, beamwidth configurations, and pruning hintsizes to measure the end-to-end performance (latency and throughput) of each ANN search system on a \stdssd\ instance; better configurations are located towards the top-left of the figures. Against Compass-in-TEE, \toolname achieves $2.3$--$12.3\times$ lower latency and $2.3$--$12.2\times$ higher throughput.
Even against the disk-optimized baselines, \toolname maintains a substantial lead: $4.0$--$4.5\times$ lower latency and $3.6$--$4.5\times$ higher throughput than RingORAM + DiskANN, and $2.5$--$9.5\times$ lower latency and $2.5$--$9.5\times$ higher throughput than PathORAM + DiskANN.
Under \highssd~configuration, additional SSD resources help the baselines narrow the gap, particularly at the high end---shrinking to $7.4\times$ against Compass-in-TEE, $6.0\times$ against PathORAM + DiskANN, and $3.9\times$ against RingORAM + DiskANN---but a significant gap remains as the higher access counts and bandwidth of the baselines carry proportionally higher compute overhead, and the workload becomes compute-bound (\appendixref{end-to-end-highssd}).

\toolname's improvement is largest where its ANN search achieves the greatest bandwidth reduction: in particular, searches over MS-MARCO benefit from a $5\times$ reduction~(\sectionref{onyx-anns-analysis}) given its highly compressed hints and large full-precision vectors.
For MS-MARCO and WIKI, RingORAM + DiskANN is the strongest baseline because these datasets use large blocks (over 3~KB) from their large embedding vectors, where bandwidth-efficient RingORAM excels.
On DEEP, PathORAM + DiskANN performs best as its small 512~B blocks favor PathORAM's lower access count.
For SIFT, Compass-in-TEE is most competitive: it uses a larger number of neighbors and higher beam width to reduce roundtrips, which is effective on small, localized graphs like 1M-vector SIFT, but becomes wasteful on larger graphs.
In all cases, \toolname achieves $2.3$--$4.5\times$ lower latency and $2.3$--$4.4\times$ higher throughput over the best baseline for each dataset.

\textbf{Throughput vs.\ recall.} \figureref{throughput-vs-recall} shows throughput-vs-recall across each dataset on a \stdssd instance.
\toolname's improvement is consistent across all recall targets: against Compass-in-TEE, we demonstrate throughput improvements of $2.1$--$13.0\times$ across datasets and recall targets; against RingORAM + DiskANN, $2.7$--$5.2\times$; and against PathORAM + DiskANN, $2.3$--$10.2\times$. The per-dataset trends mirror the latency-throughput analysis above---the same factors that favor each baseline at 90\% recall carry over to other recall targets---and in all cases, \toolname achieves $2.1$--$4.7\times$ higher throughput than the best baseline for each dataset.

\begin{table}[t]
\centering
\small
\setlength{\tabcolsep}{4pt}
\begin{tabular*}{\columnwidth}{@{\extracolsep{\fill}}l rrrr@{}}
\toprule
& \textbf{SIFT} & \textbf{MARCO} & \textbf{WIKI} & \textbf{DEEP} \\
\midrule
\multicolumn{5}{@{}l}{\textbf{Memory (GB)}} \\
\quad \toolname          & 0.17 & 0.76  & 2.39    & 2.99  \\
\quad Compass-in-TEE     & 0.16 & 0.99  & 2.93    & 3.84  \\
\quad RingORAM+DiskANN   & 0.21 & 0.73  & 2.39    & 2.57  \\
\quad PathORAM+DiskANN   & 0.13 & 0.61  & 2.16    & 2.31  \\
\midrule
\multicolumn{5}{@{}l}{\textbf{Storage (GB)}} \\
\quad \toolname          & 1.5  & 58.2  & 170.6   & 76.9  \\
\quad Compass-in-TEE     & 12.5 & 776.3 & 1361 & 706.4 \\
\quad RingORAM+DiskANN   & 5.5  & 239.7 & 515   & 271.4 \\
\quad PathORAM+DiskANN   & 1.4  & 108.2 & 208.5   & 70.7  \\
\bottomrule
\end{tabular*}
\caption{Memory and storage overhead of \toolname and baselines across our evaluation datasets.}
\label{tab:memory-usage}
\end{table}

\begin{figure*}[!t]
    \centering
    \includegraphics[width=0.70\textwidth]{plots/legend_secure_comparison.png}\\[0.2em]
    \rotatebox[origin=c]{90}{\small Queries/\$}\,
    \parbox[c]{0.96\textwidth}{\centering
        \includegraphics[width=0.23\textwidth]{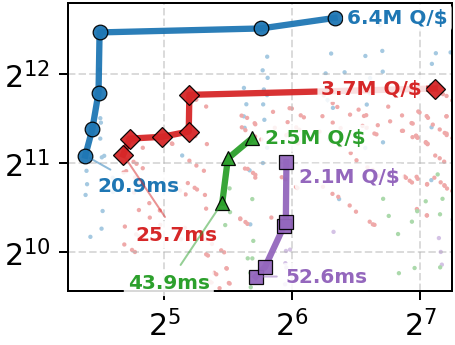} \hfill
        \includegraphics[width=0.23\textwidth]{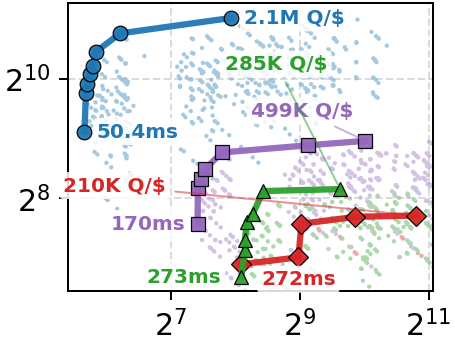} \hfill
        \includegraphics[width=0.23\textwidth]{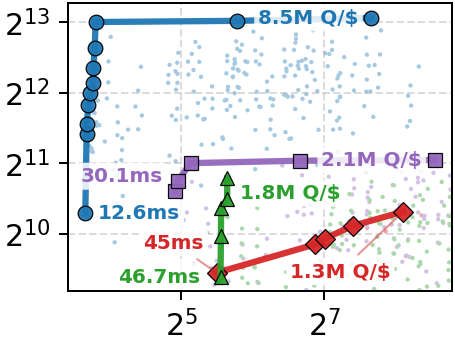} \hfill
        \includegraphics[width=0.23\textwidth]{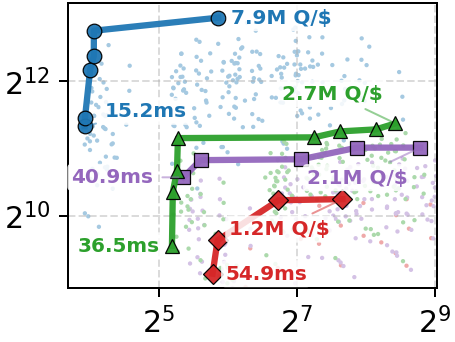}
    }\\[0.1em]
    \hspace*{0.04\textwidth}
    \parbox[t]{0.96\textwidth}{\centering\small Latency (ms)}\\[0.2em]
    \hspace*{0.04\textwidth}
    \parbox[t]{0.96\textwidth}{
        \makebox[0.23\textwidth]{\small (a) SIFT} \hfill
        \makebox[0.23\textwidth]{\small (b) MS-MARCO} \hfill
        \makebox[0.23\textwidth]{\small (c) WIKI} \hfill
        \makebox[0.23\textwidth]{\small (d) DEEP}
    }
    \caption{Pareto frontier of cost-normalized throughput (Queries/\$) vs.\ latency for \toolname and baselines, highlighting the most cost-effective configurations (vCPU and SSD resources allocated). Better configurations are towards the top-left.}
    \label{fig:throughput-dollar}
\end{figure*}

\textbf{Memory and storage.}
\tableref{memory-usage} summarizes memory and storage footprint for hosting each dataset.
PathORAM with DiskANN exhibits the lowest footprint, as it does not store any bucket metadata. \toolname uses only $10$--$30\%$ more memory, because its decoupled search architecture~(\sectionref{onyx-anns-construction}) requires two ORAM clients (one for traversal blocks, one for refinement blocks). This overhead is more noticeable on datasets with smaller hints (e.g., DEEP), where ORAM client state is a larger fraction of total memory usage.
Compass-in-TEE has the highest memory usage across most datasets because it uses RingORAM with a binary tree and stores bucket metadata locally, inflating per-block overhead without the \toolname's $d$-ary tree-based mitigation~(\sectionref{onyx-oram-construction}, \appendixref{oram-memory-storage}).
On SIFT, the index is small enough that fixed-size buffers and other constants dominate memory usage rather than metadata.

\toolname has one of the smallest storage footprints across all datasets ($2.0$--$2.7\times$ amplification over the plaintext index), with PathORAM + DiskANN achieving a comparable footprint when its binary tree capacity tightly fits $N$ (e.g., $2.3\times$ on DEEP).
RingORAM + DiskANN requires $3.0$--$4.1\times$ more storage than \toolname due to the many dummy blocks in its tree~(\sectionref{onyx-oram-construction}), and Compass-in-TEE requires $8$--$13\times$ more because its ORAM parameters have larger amplification and it uses larger blocks to accommodate CompassANN.
A single SSD unit (187.5~GB) suffices to host any of our benchmarks under \toolname.

\subsection{\toolname Co-design} \label{sec:eval-codesign}

\begin{figure}[!t]
    \centering
    \includegraphics[width=0.9\columnwidth]{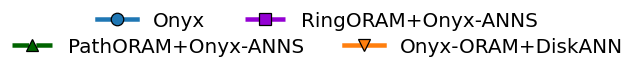}\\[0.4em]
    \rotatebox[origin=c]{90}{\footnotesize Throughput (QPS)}\,
    \parbox[c]{0.93\columnwidth}{\centering
        \includegraphics[width=0.45\columnwidth]{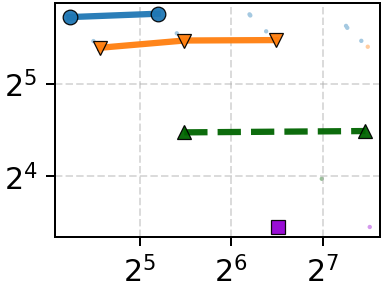}\hfill
        \includegraphics[width=0.45\columnwidth]{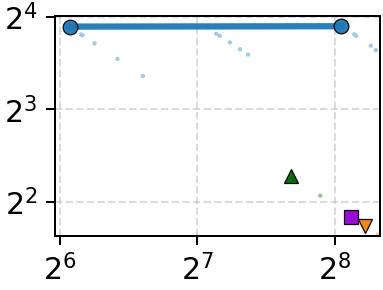}\\[-0.2em]
        \makebox[0.45\columnwidth]{\small (a) SIFT}\hfill
        \makebox[0.45\columnwidth]{\small (b) MS-MARCO}\\[0.4em]
        \includegraphics[width=0.45\columnwidth]{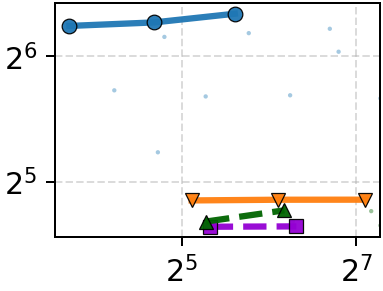}\hfill
        \includegraphics[width=0.45\columnwidth]{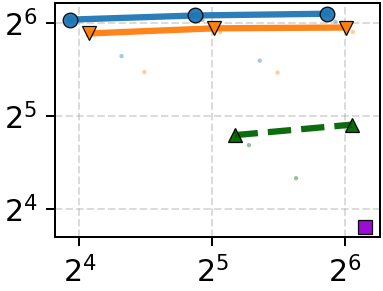}\\[-0.2em]
        \makebox[0.45\columnwidth]{\small (c) WIKI}\hfill
        \makebox[0.45\columnwidth]{\small (d) DEEP}\\[0.2em]
        {\footnotesize Latency (ms)}
    }
    \caption{Importance of ORAM-ANN co-design: pareto frontier of throughput vs. latency for \toolname and single-component (\toolname-ORAM/ANN) constructions. Better configurations are towards the top-left; single markers indicate a scheme's latency- and throughput-optimal configuration.}
    \label{fig:decomposition}
\end{figure}

\figureref{decomposition} isolates the contribution of each \toolname component by comparing against three partial combinations: \toolname-ORAM paired with DiskANN, and \toolname-ANN paired with either RingORAM or PathORAM. In every case, \toolname provides the best throughput at the lowest latency.
On the SIFT and DEEP datasets with small vectors, the ORAM is the primary bottleneck: \toolname-ORAM + Disk approaches \toolname ($10$--$20\%$ lower performance), while replacing \toolname-ORAM with RingORAM or PathORAM degrades performance by $4.9$--$5.0\times$ and $2.3$--$2.4\times$, respectively. In these datasets, \toolname-ANN's bandwidth reduction provides only a small additional benefit because DiskANN already uses small block sizes and \toolname-ORAM becomes access-count-bound.
Path + \toolname-ANNS outperforms Ring + \toolname-ANNS because PathORAM is more access-optimized, and its higher bandwidth overhead is offset by \toolname-ANN's bandwidth savings.
On MS-MARCO and WIKI which have large vectors, both components matter equally: \toolname-ORAM + Disk sees $2.8$--$4.5\times$ lower throughput than \toolname because DiskANN's large per-access bandwidth saturates the SSD, and Ring + \toolname-ANNS or Path + \toolname-ANNS show $3.2$--$4.1\times$ and $3.0$--$3.1\times$ lower throughput, respectively, because \toolname-ORAM significantly outperforms the baselines for the small 256--512 byte block accesses made by \toolname-ANNS~(\sectionref{eval-oram}).
Thus, co-design is essential: neither component alone is sufficient to achieve \toolname's performance across the full workload range.

\subsection{Cost Efficiency} \label{sec:eval-costefficiency}

Given enough CPU and disk resources, even the inefficient baselines we compare against can achieve high throughput and low latency. \figureref{throughput-dollar} demonstrates the \emph{cost-efficiency} advantage of \toolname, normalizing throughput by cost using public GCP pricing data~(\sectionref{eval-setup}).
For each scheme, we sweep over various resource allocations -- vCPU counts in $\{1, 2, 3, 4\}$ and SSD slices in $\{1, 2, 3, 4, 8, 12, 16\}$ -- and plot the Pareto-optimal queries/\$ vs latency tradeoff across all configurations.
\toolname achieves $1.7$--$9.9\times$ higher QPS/\$ than Compass-in-TEE, $3.1$--$4.2\times$ higher than Ring + Disk, and $2.6$--$7.3\times$ higher than Path + Disk.
These improvements are proportional to what we observed in \sectionref{eval-end-to-end}: to match \toolname's performance, baselines must increase both SSD and compute resources. Their poor SSD utilization---driven by higher access counts and bandwidth overheads---demands more SSD slices, and more compute to drive that I/O, so scaling resources does not meaningfully improve cost-efficiency.
\toolname's best cost-efficiency point is always at 1~SSD slice, because its compact index fits within a single slice and its throughput is not bottlenecked by the SSD.
Additionally, storage amplification locks baselines into larger configurations: Compass-in-TEE requires 4--8~SSD slices on the enterprise datasets (e.g., 8~slices on WIKI to fit the 1.4~TB index), while Ring + Disk and Path + Disk require 2--3~slices.

\section{Related Work} \label{sec:related-work}

\paragraph{Private nearest neighbor search.}
A number of cryptographic approaches target private nearest neighbor search with strong security guarantees, using techniques based on ORAM~\cite{compass}, MPC~\cite{sanns}, homomorphic encryption~\cite{hers, panther}, and PIR~\cite{tiptoe, pacmann}.
These approaches incur high overhead due to the strong cryptographic security model they target; among them, Compass~\cite{compass} is the state-of-the-art, but still requires $>$1~s latency and roughly 100 queries per dollar.
Unlike these works, \toolname relies on hardware trust assumptions (TEEs) to offer a practical solution with low latency and high cost-efficiency, achieving $12$~ms latency and over $8$ million queries per dollar~(\sectionref{eval-end-to-end}).

\textbf{Oblivious RAM.}
ORAM was first introduced by the seminal work of Goldreich and Ostrovsky~\cite{goldreich-ostrovsky} in 1996.
Since then, a long line of work has proposed ORAM constructions, including hierarchical schemes~\cite{goldreich-ostrovsky, goodrich-oram, pinkas-oram, panorama, futorama} and the more popular tree-based designs~\cite{pathoram, ringoram, circuit-oram}.
ORAMs have been applied to various settings, including the traditional client-server model~\cite{ringoram, obladi, privatefs, oblivistore, taostore, compass, treebeard}, hardware enclaves~\cite{oblix, oblidb, graphos, zerotrace, snoopy, enigmap, h2o2ram}, and multi-server settings~\cite{abraham-pir-oram, metal, mskt-oram, doram, floram, scoram}.
\toolname operates in the enclave or TEE setting and is the first to optimize ORAM for hiding ANN search disk access patterns.
The closest ORAM work to ours is EnigMap~\cite{enigmap}, which focuses on doubly oblivious access to an external storage from a TEE, where the ORAM client itself is data-oblivious to hide memory access patterns within the TEE.
\toolname-ORAM achieves better performance than EnigMap in our (singly-oblivious) setting by leveraging $d$-ary trees, which were first introduced by Gentry et al.~\cite{gentry-oram}. $d$-ary trees are also used in multi-server settings~\cite{abraham-pir-oram, mskt-oram} where their high bandwidth overhead is offset through private information retrieval.
We focus on tree-based ORAM, but our access-optimized design is also compatible with hierarchical ORAMs~\cite{panorama, futorama} since they also use bandwidth-efficient large bucket designs like RingORAM.

\textbf{Disk-based ANN search.}
Disk-based ANN search includes both graph-based~\cite{diskann} and clustering-based~\cite{faiss, spann, pq-quantized-hints} approaches.
We focus on graph-based designs, in particular DiskANN~\cite{diskann}, which is widely deployed in industry~\cite{azure-cosmos-db-vector, pgvectorscale, jvector, intel-svs, zilliz-milvus, pinecone, diskann-overview} and has been extended to filtered search~\cite{filtered-diskann}, streaming updates~\cite{fresh-diskann, inplace-diskann}, low-memory operation~\cite{lm-diskann}, parallel indexing~\cite{parlayann}, and distributed indices~\cite{distributedann}.
These designs optimize for the plain-SSD regime where bandwidth is cheap and access count is the bottleneck; for instance, Starling~\cite{starling} reduces access count by $2\times$ while increasing bandwidth by up to $16\times$.
\toolname-ANNS targets bandwidth-efficient disk-based ANN search, motivated by ORAM's $O(\log N)\times$ bandwidth amplification~(\sectionref{onyx-anns}).
Bandwidth efficiency has also been explored for in-memory ANN~\cite{membw-ann}, which reduces memory bandwidth through incremental reads and early rejection of candidates.
These techniques are incompatible with our setting because they require fine-grained, data-dependent accesses per candidate.
We focus on graph-based ANN search in this work but our refinement strategy can also be used in clustering-based indices~\cite{faiss, spann} to minimize bandwidth.

\section*{Acknowledgments}
We thank Benjamin Karsin and Vikram Sharma Mailthody for their helpful discussions and insightful comments that helped shape this work.
We are also grateful to the Sky Lab security group members for their valuable feedback.
This work was supported in part by NSF CAREER Award 1943347, and by generous gifts from Accenture, AMD, Anyscale, Broadcom, Cisco, IBM, Intel, Intesa Sanpaolo, Lambda, Lightspeed, Mibura, Microsoft, NVIDIA, Samsung SDS, and SAP.

\bibliographystyle{ACM-Reference-Format}
\bibliography{references}

\appendix
\section{Additional Evaluation Results} \label{app:additional-eval} \label{app:oram-highssd} \label{app:recall-highssd} \label{app:decomp-highssd} \label{app:ann-ablation-other}

This appendix contains additional evaluation results, including the \highssd\ variants of the main-body figures and the ablation studies.

\subsection{\toolname-ORAM (\highssd)}
\figureref{oram-eval-highssd} complements \figureref{oram-eval}. The trends are consistent: \toolname-ORAM's improvements carry over with additional SSD resources. At 512~B, \toolname-ORAM achieves $3.9\times$ higher throughput and $2.7\times$ lower latency than RingORAM, and $2.9\times$ higher throughput and $2.0\times$ lower latency than PathORAM. The throughput gains are smaller than under \stdssd\ ($5.1\times$ and $3.4\times$) because the baselines were I/O-bottlenecked and benefit from the extra SSD resources; the gap nevertheless persists because they become compute-bound~(\sectionref{eval-end-to-end}).

\begin{figure}[H]
    \centering
    \includegraphics[width=1.02\linewidth]{plots/legend_oram_bench_ablation.png}\\[0.2em]
    \includegraphics[width=0.49\linewidth]{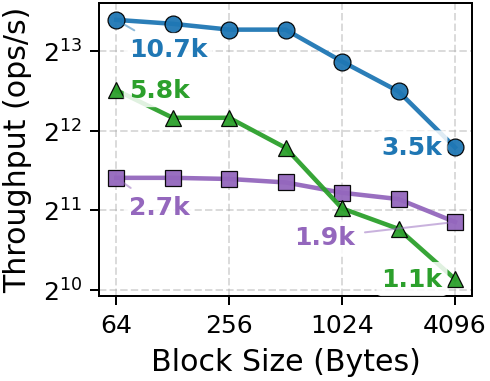} \hfill
    \includegraphics[width=0.49\linewidth]{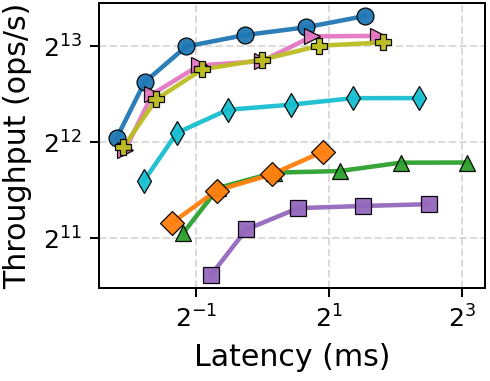}\\[0.2em]
    \makebox[0.49\linewidth]{\small (a) Throughput vs.\ block size} \hfill
    \makebox[0.49\linewidth]{\small (b) Throughput vs.\ latency ($B\!=\!512$~B)}
    \caption{\toolname-ORAM vs prior work for $20$M blocks (\highssd). LBM (local bucket metadata), FBR (full bucket reads), ST (shallow tree) refer to steps in \sectionref{onyx-oram-construction}; \toolname: Ring+LBM+FBR+ST-8; Ring: RingORAM; Path: locality-optimized PathORAM.}
    \label{fig:oram-eval-highssd}
\end{figure}

\subsection{ORAM Memory and Storage} \label{app:oram-memory-storage}
We measure the per-block memory footprint of \toolname-ORAM, PathORAM, RingORAM, and RingORAM with local bucket metadata (LBM) in our implementation.
\toolname-ORAM ($d = 8$) requires ${\sim}8.7$~bytes/block, comparable to RingORAM at ${\sim}8.6$~bytes/block and only slightly above PathORAM at ${\sim}7$ bytes/block.
Adding LBM to RingORAM, however, raises the cost to ${\sim}17.6$~bytes/block, demonstrating that the $d$-ary tree is essential to keep local metadata overhead low (\sectionref{onyx-oram-construction}).

\toolname-ORAM also achieves better storage amplification than both baselines across all evaluated datasets, only requiring $1.7$--$1.8\times$ the disk footprint of the original plaintext dataset, compared to $2.1$--$3.8\times$ for PathORAM and $7.5$--$7.7\times$ for RingORAM. Storage amplification costs can vary across datasets, as $N$ may not tightly fit the ORAM tree capacity, requiring an additional level that is largely underutilized.

\subsection{End-to-End Oblivious ANN Search} \label{app:end-to-end-highssd}
\paragraph{Throughput vs.\ latency.}
\figureref{end-to-end-highssd} complements \figureref{throughput-vs-latency} under the \highssd\ configuration.
\toolname achieves $2.2$--$7.4\times$ lower latency and $2.2$--$7.3\times$ higher throughput against Compass-in-TEE; $3.2$--$4.6\times$ lower latency and $3.2$--$3.9\times$ higher throughput against RingORAM + DiskANN; and $2.4$--$6.0\times$ lower latency and $2.3$--$6.0\times$ higher throughput against PathORAM + DiskANN.

\paragraph{Throughput vs.\ recall.}
\figureref{throughput-vs-recall-highssd} shows the throughput-vs-recall results under \highssd.
The trends are similarly consistent: $2.1$--$7.3\times$ against Compass-in-TEE, $2.6$--$4.0\times$ against RingORAM + DiskANN, and $2.2$--$6.0\times$ against PathORAM + DiskANN across all datasets and recall targets.

\begin{figure*}[t]
    \centering
    \includegraphics[width=0.80\textwidth]{plots/legend_secure_comparison.png}\\[0.2em]
    \rotatebox[origin=c]{90}{\small Throughput (QPS)}\,
    \parbox[c]{0.96\textwidth}{\centering
        \includegraphics[width=0.23\textwidth]{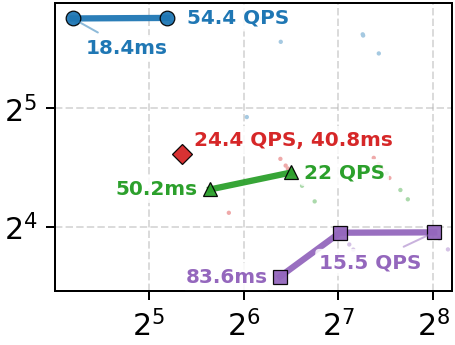} \hfill
        \includegraphics[width=0.23\textwidth]{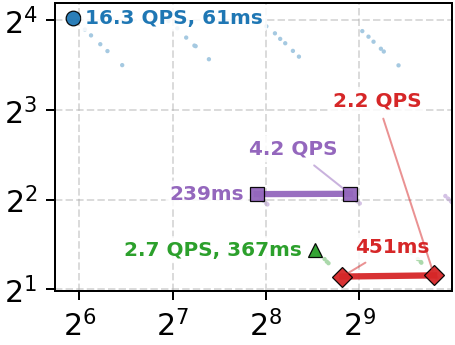} \hfill
        \includegraphics[width=0.23\textwidth]{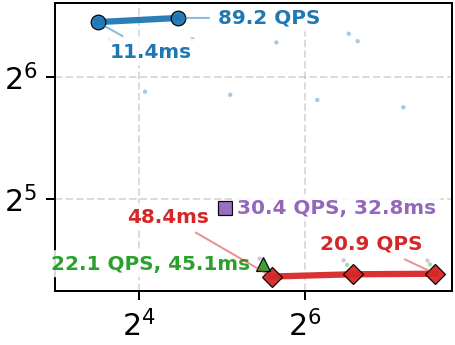} \hfill
        \includegraphics[width=0.23\textwidth]{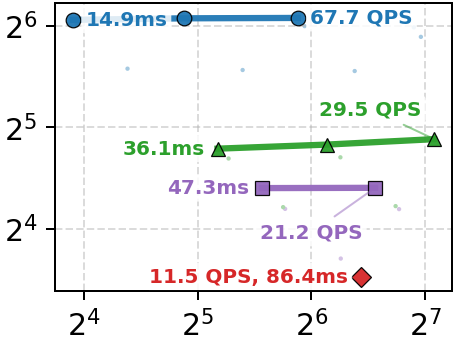}
    }\\[0.1em]
    \hspace*{0.04\textwidth}
    \parbox[t]{0.96\textwidth}{\centering\small Latency (ms)}\\[0.3em]
    \rotatebox[origin=c]{90}{\small Throughput (QPS)}\,
    \parbox[c]{0.96\textwidth}{\centering
        \includegraphics[width=0.23\textwidth]{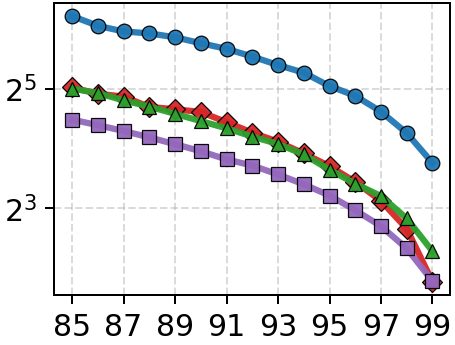} \hfill
        \includegraphics[width=0.23\textwidth]{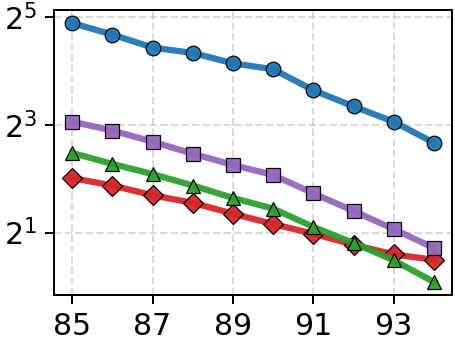} \hfill
        \includegraphics[width=0.23\textwidth]{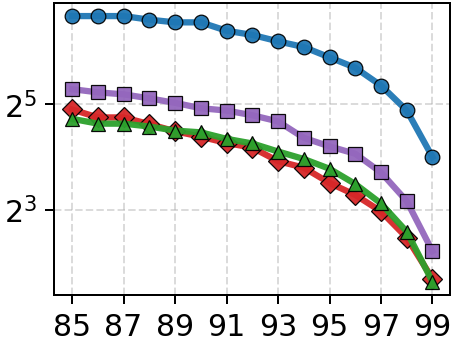} \hfill
        \includegraphics[width=0.23\textwidth]{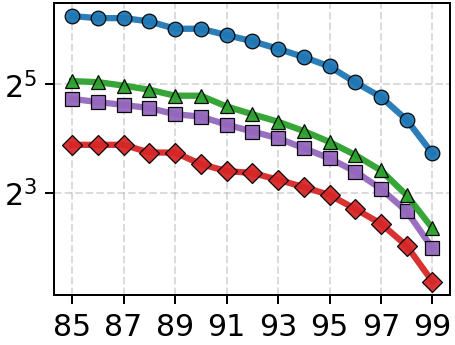}
    }\\[0.1em]
    \hspace*{0.04\textwidth}
    \parbox[t]{0.96\textwidth}{\centering\small Recall}\\[0.2em]
    \hspace*{0.04\textwidth}
    \parbox[t]{0.96\textwidth}{
        \makebox[0.23\textwidth]{\small (a) SIFT} \hfill
        \makebox[0.23\textwidth]{\small (b) MARCO} \hfill
        \makebox[0.23\textwidth]{\small (c) WIKI} \hfill
        \makebox[0.23\textwidth]{\small (d) DEEP}
    }
    \caption{(Top) Pareto frontier of throughput vs.\ latency of \toolname and baselines (\highssd; top-left is better). Dots below the curves represent sub-optimal parameterizations; single markers are the configuration for a scheme that maximizes both latency and throughput. (Bottom) Throughput vs.\ recall of \toolname and baselines (\highssd).}
    \label{fig:end-to-end-highssd}
    \label{fig:throughput-vs-recall-highssd}
\end{figure*}

\subsection{Co-design Ablation (\highssd)} \label{app:decomp-highssd-section}
\figureref{decomposition-highssd} complements \figureref{decomposition}. The same co-design insights from \sectionref{eval-codesign} apply: on small-vector datasets (SIFT, DEEP), the ORAM optimization remains the primary driver, while on large-vector datasets (MS-MARCO, WIKI) both components contribute. The gaps are modestly smaller under \highssd\ as additional SSD resources reduce the I/O bottleneck.

\begin{figure}[t]
    \centering
    \includegraphics[width=0.9\columnwidth]{plots/legend_decomposition.png}\\[0.4em]
    \rotatebox[origin=c]{90}{\footnotesize Throughput (QPS)}\,
    \parbox[c]{0.93\columnwidth}{\centering
        \includegraphics[width=0.45\columnwidth]{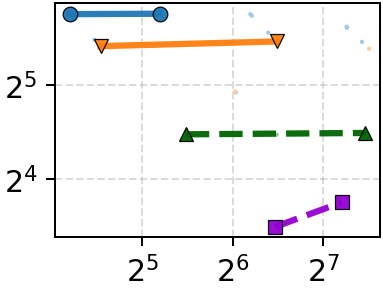}\hfill
        \includegraphics[width=0.45\columnwidth]{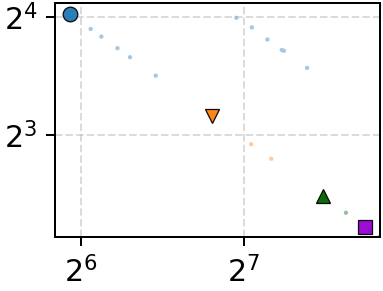}\\[-0.2em]
        \makebox[0.45\columnwidth]{\small (a) SIFT}\hfill
        \makebox[0.45\columnwidth]{\small (b) MS-MARCO}\\[0.4em]
        \includegraphics[width=0.45\columnwidth]{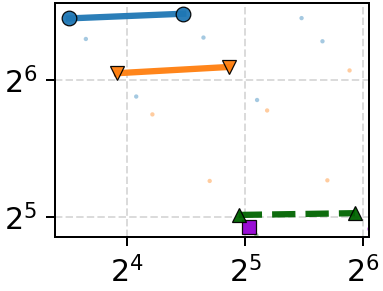}\hfill
        \includegraphics[width=0.45\columnwidth]{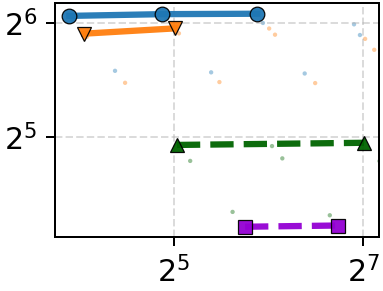}\\[-0.2em]
        \makebox[0.45\columnwidth]{\small (c) WIKI}\hfill
        \makebox[0.45\columnwidth]{\small (d) DEEP}\\[0.2em]
        {\footnotesize Latency (ms)}
    }
    \caption{ORAM-ANN co-design ablation (\highssd). Pareto frontier of throughput vs.\ latency for \toolname and single-component (\toolname-ORAM and \toolname-ANN) constructions. Better configurations are towards the top-left.}
    \label{fig:decomposition-highssd}
\end{figure}

\subsection{\toolname-ANNS Ablation} \label{app:ann-ablation}

\figureref{ablation-ann} compares the throughput vs.\ latency of \toolname-ANNS with various pruning hint sizes, DiskANN (DA), and naive decoupling when each is combined with \toolname-ORAM.

Naive decoupling, which uses small in-memory traversal hints to prune the candidates list, fails to improve over DiskANN on three of four datasets: it is $25$--$27\%$ slower on SIFT and DEEP, and roughly matches DiskANN on MS-MARCO.
The in-memory hints are too coarse to effectively prune candidates, barely reducing bandwidth ($1.0$--$1.3\times$, \tableref{onyx_ann_vs_diskann_k10_r95}) while nearly doubling the access count ($+72$--$99\%$).
Only on WIKI, where vectors are large enough to yield a $2.2\times$ bandwidth reduction, does naive decoupling outperform DiskANN ($1.7\times$).

\toolname-ANNS resolves this by introducing pruning hints that progressively reduce bandwidth without significantly increasing access count ($+5$--$15\%$). As pruning hint size increases, bandwidth savings grow and throughput improves: for instance, on MS-MARCO, increasing hints from 64~B to 256~B improves throughput from 4.8 to 14.9~QPS ($4.5\times$ over DiskANN). The effect is most pronounced on datasets with large vectors (MS-MARCO $4.5\times$, WIKI $2.8\times$) where the bandwidth reduction is largest ($3.8$--$5.0\times$, \sectionref{onyx-anns-analysis}), but \toolname-ANNS also improves on smaller-vector datasets (SIFT $1.2\times$, DEEP $1.1\times$). Without increasing memory footprint beyond what naive decoupling or DiskANN requires, \toolname-ANNS consistently achieves the best throughput-latency tradeoff.

\begin{figure*}[t]
    \centering
    \includegraphics[width=0.40\textwidth]{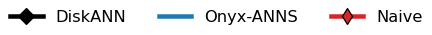}\\[0.2em]
    \rotatebox[origin=c]{90}{\small Throughput (QPS)}\,
    \parbox[c]{0.96\textwidth}{\centering
        \includegraphics[width=0.23\textwidth]{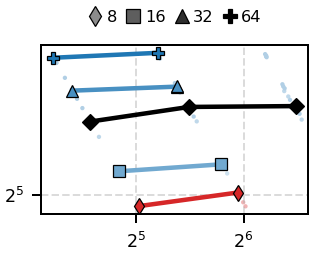} \hfill
        \includegraphics[width=0.23\textwidth]{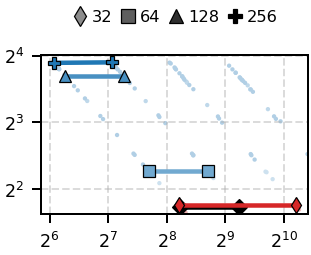} \hfill
        \includegraphics[width=0.23\textwidth]{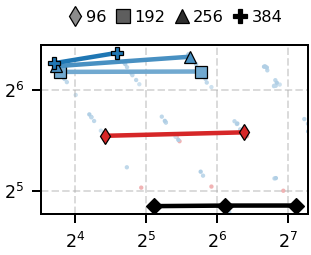} \hfill
        \includegraphics[width=0.23\textwidth]{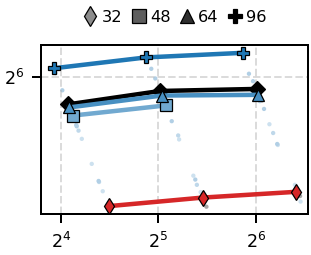}
    }\\[0.1em]
    \hspace*{0.04\textwidth}
    \parbox[t]{0.96\textwidth}{\centering\small Latency (ms)}\\[0.2em]
    \hspace*{0.04\textwidth}
    \parbox[t]{0.96\textwidth}{
        \makebox[0.23\textwidth]{\small (a) SIFT} \hfill
        \makebox[0.23\textwidth]{\small (b) MS-MARCO} \hfill
        \makebox[0.23\textwidth]{\small (c) WIKI} \hfill
        \makebox[0.23\textwidth]{\small (d) DEEP}
    }
    \caption{\toolname-ANNS ablation (\stdssd). Each curve varies the search queue depth (smallest QD value); other labeled values are pruning hint sizes. Pareto frontiers highlight the best-performing configurations for each approach.}
    \label{fig:ablation-ann}
\end{figure*}

\newcommand{\efconstruction}{\texttt{efConstruction}}
\section{Index Construction Hyperparameters} \label{app:index-hyperparams}

\textbf{\toolname Index (DiskANN).}
\toolname builds on DiskANN~\cite{diskann}, which constructs a Vamana proximity graph on disk.
All indices use a build search list size of $L = 128$, L2 distance, and 32-bit floating-point vectors.
The max graph degree is $R = 32$ for all datasets except MS-MARCO, which uses $R = 64$.
For pruning hints, we sweep over multiple hint sizes per dataset and select the best configuration: SIFT uses 16, 32, 64~B; MS-MARCO uses 64, 128, 256~B; WIKI uses 192, 256, 384~B; and DEEP uses 48, 64, 96~B.

\textbf{Compass Index (HNSW + PQ).}
Compass indices were built using the authors' Faiss fork~\cite{compass-faiss-fork}, which constructs an HNSW graph for traversal and uses the same PQ codes~\cite{pq-quantized-hints} as DiskANN for in-memory traversal hints.
For Compass datasets, we used the same parameters as the paper~\cite{compass}: $M=64$ and $\efconstruction = 80$ for SIFT-1M and $M=128$ and $\efconstruction=200$ for MS-MARCO.
For our benchmarks, we use $\efconstruction = 200$ and $M=64$, which means the maximum graph degree is $2M = 128$.

\section{Dynamic ANN Index} \label{app:dynamic-index}
In this section, we discuss how \toolname-ANNS's decoupled index layout can be applied to FreshDiskANN~\cite{fresh-diskann} to support dynamic indices with insertions and deletions.

\paragraph{FreshDiskANN Update Primitives.}
FreshDiskANN's dynamic updates are built from two primitives: $\GreedySearch$ and $\RobustPrune$. At a high level, $\GreedySearch$ traverses the graph to find a candidate list of nearby nodes, while $\RobustPrune$ takes a target node together with such a candidate list and selects which candidates should be neighbors of the target.
Concretely, $\RobustPrune$ processes candidates in increasing order of approximate distance (using in-memory traversal hints) to the target vector $\query$. When it keeps a candidate $\point$, it removes any later candidate $\point'$ that is already well covered by $\point$, i.e., whenever $\dist(\pointvec,\embvec{\point'}) \leq \alpha \cdot \dist(\queryvec,\embvec{\point'})$. Here, $\alpha > 1$ is the usual slack parameter that improves graph connectivity and convergence~\cite{fresh-diskann}.

\textbf{Insertions.}
With these primitives, insertion is straightforward in Fresh-DiskANN. For a new point, we first run $\GreedySearch$ to find nearby nodes, and then run $\RobustPrune$ to choose its outgoing neighbors while maintaining the maximum out-degree. We then try to add reverse edges from those selected neighbors back to the new point by rerunning $\RobustPrune$ on each affected neighborhood.

\textbf{Deletions.}
Deletions are trickier because every node that points to the deleted point must be updated, and there may be many such nodes. As in prior dynamic graph indices, this repair is handled lazily and in batches. In the background, the system streams through the dataset, identifies points whose neighborhoods contain deleted nodes, uses the neighbors of the deleted nodes as repair candidates, and reruns $\RobustPrune$ to rebuild those neighborhoods. Batching amortizes the cost of this full streaming pass over many deletions.

\paragraph{Cost benefit of \toolname-ANNS layout.}
In both cases, the external-storage cost is heavily dominated by fetching neighbor information for the relevant candidates. With the \toolname-ANNS decoupled layout, this requires fetching only adjacency lists together with pruning hints, which fit in much smaller blocks. In contrast, in coupled designs such as DiskANN and Compass, the neighbor information is stored together with the full-precision vectors, so accessing only the graph metadata still requires fetching the entire larger block.
Concretely, for WIKI-20M, this translates to performing an equal number of accesses but with $384$-byte blocks as opposed to $3.2$~KB blocks with the coupled baseline.

\textbf{Writeback via eviction.}
The updated neighborhoods also need to be written back to external storage. This can naturally be piggybacked on the ORAM eviction procedure, as ORAM necessitates that every block that is accessed must be written back.

\section{\toolname-ORAM Protocol}\label{app:oram-protocol}
The \toolname-ORAM protocol is summarized in \algorithmref{oram-protocol}. It relies on read helpers in \algorithmref{oram-helpers} and write helpers in \algorithmref{oram-write-helpers}.

\paragraph{Notation.}
The ORAM tree has arity $d$, depth $L$, bucket capacity $Z$, dummy slots $S$, and eviction frequency $A$ (one eviction per $A$ accesses).
$P(\ell, i)$ denotes the bucket at level $i$ on the root-to-leaf path to leaf $\ell$.
$\stash$ is the client-local stash of blocks not yet written back, and $G$ is the global access counter.
$\textsc{ReverseDigits}_{d,L}(x)$ reverses the $L$-digit base-$d$ representation of $x$, yielding the reverse-lexicographic eviction order.

\paragraph{Position map.}
The position map $\posmap$ maps each block address $a$ to a triple $(\pmleaf, \pmlvl, \pmslot)$:
$\posmap[a].\pmleaf$ is the assigned leaf,
$\posmap[a].\pmlvl$ is the tree level of the bucket currently holding block $a$ (or $\bot$ if the block is in the stash), and
$\posmap[a].\pmslot$ is the logical slot index within that bucket (or $\bot$ if in the stash).
The cached $(\pmlvl, \pmslot)$ allow \textsc{ReadPath} to locate the target block directly: at the matching level it reads the known slot, and at all other levels it reads the next unread dummy.

\paragraph{Bucket layout on disk.}
Each bucket $b$ on disk consists of $Z + S$ individually encrypted slots, denoted $\text{disk}[b][0..Z{+}S{-}1]$.
Each slot encrypts a pair $(\oramaddrs, \oramdata)$: for real blocks, $\oramaddrs$ is the block address; for dummies, $\oramaddrs = \bot$.
Decrypting a slot yields both the address and the data.
During eviction, the entire bucket is read or written as a single I/O; during a regular access, \textsc{ReadBlock} fetches a single block.

\paragraph{Per-bucket local metadata.}
Each bucket $b$ has metadata $\metadata[b]$ stored locally at the ORAM client, consisting of:
\begin{itemize}[leftmargin=*]
\item A monotonic version counter $\oramver$.
\item A pseudorandom permutation $\oramptrs[0..Z{+}S{-}1]$ mapping logical slot indices to physical slot offsets on disk. Logical slots $0..Z{-}1$ correspond to real block slots and $Z..Z{+}S{-}1$ to dummy slots.
\item Per-slot validity bits $\oramvalid[0..Z{+}S{-}1]$, all set to $1$ after a write. A bit is cleared when the corresponding slot is read.
\item A dummy counter $\dummycount$, initialized to $0$ after each write. The next dummy slot is at physical offset $\oramptrs[Z + \dummycount]$.
\item A total access counter $\oramcount$, initialized to $0$ after each write. An early reshuffle is triggered when $\oramcount$ reaches $S$ (i.e., this bucket has been accessed $S$ times).
\end{itemize}

\paragraph{Authenticated encryption.}
All slots are protected by an authenticated encryption (AE) scheme with associated data (instantiated with AES-GCM~\cite{aes-gcm}) using a global secret key $\oramsk$.
We write $\Enc(b, k, x)$ and $\Dec(b, k, \mathsf{ct})$ for AE encryption and decryption of physical slot $k \in [Z{+}S]$ in bucket $b$.
The associated data for AEAD is $\oramaad = (b \,\|\, k \,\|\, \metadata[b].\oramver)$, binding each ciphertext to its bucket index~$b$, physical slot offset~$k$, and the bucket's current version~$\oramver$.
This prevents the adversary from replaying, reordering, or substituting ciphertexts across slots, buckets, or versions.
Every slot read from disk---whether it contains a real block or a dummy---is integrity-verified via the AEAD tag; this ensures the adversary cannot cause selective failures that depend on the block's identity.
If any integrity check fails, the protocol aborts immediately.

\begin{algorithm}[htb]
\caption{\toolname-ORAM Protocol}
\label{alg:oram-protocol}
\begin{algorithmic}[1]
\Statex \hspace{-1.2em}\textbf{Parameters:} Tree-arity $d$, bucket capacity $Z$, dummy slots $S$, eviction frequency $A$
\Statex \hspace{-1.2em}\textbf{Client state:} Position map $\posmap$, stash $\stash$, per-bucket metadata $\metadata$, round counter $G$
\Statex \hspace{-1.2em}\textbf{Server state:} Encrypted $d$-ary tree of depth $L$ with buckets of $Z + S$ slots
\Statex

\Function{Access}{$a, \oramop, \oramdata'$}
    \State $(\pmleaf_a,\, \pmlvl_a,\, \pmslot_a) \gets \posmap[a]$
    \State $\pmleaf' \gets \text{UniformRandom}(0, d^L - 1)$
    \State $\posmap[a] \gets (\pmleaf',\, \bot,\, \bot)$
    \State $\oramdata \gets \textsc{ReadPath}(\pmleaf_a, \pmlvl_a, \pmslot_a)$
    \If{$\oramdata = \bot$}
        $\oramdata \gets$ remove $a$ from $\stash$
    \EndIf
    \If{$\oramop = \textsc{write}$}
        $\oramdata \gets \oramdata'$
    \EndIf
    \State $\stash \gets \stash \cup \{(a, \pmleaf', \oramdata)\}$
    \State $G \gets G + 1$
    \If{$G \bmod A = 0$}
        \textsc{EvictPath}()
    \EndIf
    \State \textsc{EarlyReshuffle}($\pmleaf_a$)
    \If{$\oramop = \textsc{read}$}
        \Return $\oramdata$
    \EndIf
\EndFunction

\Statex

\Function{ReadPath}{$\pmleaf_a, \pmlvl_a, \pmslot_a$}
    \State $\oramdata \gets \bot$
    \For{$i \gets 0$ to $L$}
        \If{$i = \pmlvl_a$}
            $j \gets \pmslot_a$
        \Else\; $j \gets \bot$
        \EndIf
        \State $\oramoffset \gets \textsc{GetOffset}(P(\pmleaf_a, i), j)$
        \State $\oramdata' \gets \textsc{ReadBlock}(P(\pmleaf_a, i), \oramoffset)$
        \If{$\oramdata' \neq \bot$}
            $\oramdata \gets \oramdata'$
        \EndIf
    \EndFor
    \State \Return $\oramdata$
\EndFunction

\Statex

\Function{EvictPath}{}
    \State $\ell_e \gets \textsc{ReverseDigits}_{d,L}(G/A \bmod d^L)$
    \For{$i \gets 0$ to $L$}
        \State $\stash \gets \stash \cup \textsc{ReadBucket}(P(\ell_e, i))$
    \EndFor
    \For{$i \gets L$ to $0$}
        \State \textsc{WriteBucket}($P(\ell_e, i)$, $\stash$)
    \EndFor
\EndFunction

\end{algorithmic}
\end{algorithm}

\begin{algorithm}[htb]
\caption{\toolname-ORAM Read Helpers}
\label{alg:oram-helpers}
\begin{algorithmic}[1]

\Function{GetOffset}{$b, j$}
    \If{$j = \bot$}
        \State $j \gets Z + \metadata[b].\dummycount$
        \State $\metadata[b].\dummycount \mathrel{+}= 1$
    \Else
    \EndIf
    \State $\metadata[b].\oramcount \mathrel{+}= 1$
    \State $k \gets \metadata[b].\oramptrs[j]$
    \State $\metadata[b].\oramvalid[j] \gets 0$
    \State \Return $k$
\EndFunction

\Statex

\Function{ReadBlock}{$b, \oramoffset$}
    \State \textbf{read block} $\text{disk}[b][\oramoffset]$ \textbf{from disk}
    \State $(\_, \oramdata) \gets \Dec(b, \oramoffset, \text{disk}[b][\oramoffset])$
    \If{decryption fails} \textbf{abort} \EndIf
    \State \Return $\oramdata$
\EndFunction

\Statex

\Function{ReadBucket}{$b$}
    \State \textbf{read bucket} $\text{disk}[b]$ \textbf{from disk}
    \For{$j \gets 0$ to $Z + S - 1$}
        \State $k \gets \metadata[b].\oramptrs[j]$
        \State $(\oramaddrs[j],\, \oramdata[j]) \gets \Dec(b, k, \text{disk}[b][j])$
        \If{decryption fails} \textbf{abort} \EndIf
    \EndFor
    \State $\mathsf{blocks} \gets \emptyset$
    \For{$j \gets 0$ to $Z - 1$}
        \State $k \gets \metadata[b].\oramptrs[j]$
        \If{$\oramaddrs[j] \neq \bot$ and $\metadata[b].\oramvalid[j]$}
            \State $a \gets \oramaddrs[j]$
            \State $\pmleaf_a \gets \posmap[a].\pmleaf$
            \State $\mathsf{blocks} \gets \mathsf{blocks} \cup \{(a, \pmleaf_a, \oramdata[j])\}$
        \EndIf
    \EndFor
    \State \Return $\mathsf{blocks}$
\EndFunction

\end{algorithmic}
\end{algorithm}

\begin{algorithm}[htb]
\caption{\toolname-ORAM Write Helpers}
\label{alg:oram-write-helpers}
\begin{algorithmic}[1]

\Function{WriteBucket}{$b, \stash$}
    \State select up to $Z$ blocks from $\stash$ assignable to $b$
    \State let $(a_j, \pmleaf_j, \oramdata_j)_{j=0}^{Z'-1}$ be the selected blocks
    \State remove selected blocks from $\stash$
    \State $\metadata[b].\oramptrs \gets \textsc{RandomPerm}(0, Z{+}S{-}1)$
    \State $\metadata[b].\oramver \mathrel{+}= 1$
    \For{$j \gets 0$ to $Z + S - 1$}
        \State $k \gets \metadata[b].\oramptrs[j]$
        \If{$j < Z'$}
            \State $\text{disk}[b][k] \gets \Enc(b, k, (a_j, \oramdata_j))$
        \Else
            \State $\text{disk}[b][k] \gets \Enc(b, k, (\bot, \bot))$
        \EndIf
    \EndFor
    \State $\metadata[b].\oramvalid \gets \{1\}^{Z+S}$
    \State $\metadata[b].\oramcount \gets 0$
    \State $\metadata[b].\dummycount \gets 0$
    \For{$j \gets 0$ to $Z' - 1$}
        \State $\posmap[a_j].(\pmlvl, \pmslot) \gets (\text{level}(b),\; j)$
    \EndFor
    \State \textbf{write bucket} $\text{disk}[b]$ \textbf{to disk}
\EndFunction

\Statex

\Function{EarlyReshuffle}{$\pmleaf$}
    \For{$i \gets 0$ to $L$}
        \If{$\metadata[P(\pmleaf, i)].\oramcount = S$}
            \State $\stash \gets \stash \cup \textsc{ReadBucket}(P(\pmleaf, i))$
            \State \textsc{WriteBucket}($P(\pmleaf, i)$, $\stash$)
        \EndIf
    \EndFor
\EndFunction

\end{algorithmic}
\end{algorithm}

\newcommand{\oramZAL}{\text{ORAM}^{Z,A}_L}
\newcommand{\oramInfAL}{\text{ORAM}^{\infty,A}_L}
\newcommand{\stashst}{\mathit{st}(S_Z)}
\newcommand{\nT}{n(T)}
\newcommand{\cT}{c(T)}
\newcommand{\XT}{X(T)}
\newcommand{\Yb}{Y(b)}

\section{Eviction Analysis for $d$-ary RingORAM}\label{app:eviction-proof}

This section proves \theoremref{dary-stash}.
The analysis builds directly on the stash analysis of RingORAM~\cite{ringoram} (Section~4 of that work).
We adopt the same notation: $\oramZAL$ denotes a non-recursive Ring ORAM with $L+1$ levels, bucket size $Z$, and one eviction per $A$ accesses, and $\oramInfAL$ denotes the corresponding $\infty$-ORAM with infinite bucket capacity.
The only structural difference is that the tree is $d$-ary instead of binary: each internal node has $d$ children, and the tree has $d^L$ leaves.
Eviction paths cycle through all $d^L$ leaves in reverse-lexicographic order.

The proof follows the same two-step structure as RingORAM.
The first step (Lemmas~1 and~2 of~\cite{ringoram}) establishes that the stash overflow probability of $\oramZAL$ can be bounded by the probability that any rooted subtree in $\oramInfAL$ is overloaded.
These lemmas depend only on the tree structure and the greedy eviction algorithm, and hold unchanged for $d$-ary trees.
The only difference is in the combinatorial bound on the number of rooted subtrees with $n$ nodes.
In a $d$-ary tree, this count equals the $d$-ary Catalan number $\frac{1}{n}\binom{dn}{n-1}$.
Using $\binom{dn}{n-1} \leq \binom{dn}{n} \leq (ed)^n$ (the last step via $\binom{m}{k} \leq (em/k)^k$), the subtree count is at most $(ed)^n$.
This reduces to $(2e)^n$ for $d = 2$, which is a slightly looser but simpler bound than the $4^n$ used in~\cite{ringoram}.
Concretely:
\[
    \Pr[\stashst > R] \;\leq\; \sum_{n \geq 1} (ed)^n \max_{T: \nT = n} \Pr[\XT > \cT + R],
\]
where $T$ ranges over rooted subtrees of $\oramInfAL$, $\nT$ is the number of nodes, $\cT = \nT \cdot Z$ is the capacity, and $\XT$ is the number of blocks in $T$ before post-processing.

The second step is where the generalization to $d$-ary trees matters: we must bound the expected bucket load $E[\Yb]$ for each bucket $b$ in $\oramInfAL$ before post-processing.

\begin{lemma}[Expected bucket load in $d$-ary RingORAM]\label{lem:dary-load}
For any bucket $b \in \oramInfAL$ with a $d$-ary tree, if $N \leq \frac{A(d-1)}{2} \cdot d^L$, then $E[\Yb] \leq \frac{A(d-1)}{2}$.
\end{lemma}

\begin{proof}
\textbf{Leaf buckets.}
A leaf bucket $b$ contains blocks placed there by the last \textsc{EvictPath} through that leaf.
There are at most $N$ distinct blocks, each mapped to $b$ independently with probability $d^{-L}$.
Thus $E[\Yb] \leq N \cdot d^{-L} \leq \frac{A(d-1)}{2}$.

\textbf{Non-leaf buckets.}
Let $b$ be a bucket at level $i$ ($0 \leq i < L$) with $d$ children $c_1, \ldots, c_d$.
Let $m_1 < m_2 < \cdots < m_d$ be the times of the last \textsc{EvictPath} operation through each child.
Due to the deterministic reverse-lexicographic eviction order, a bucket at level $i$ is evicted every $d^{i}$ eviction operations, and consecutive children of $b$ are evicted $d^i$ eviction operations apart.
Therefore $m_d - m_1 = (d-1) \cdot d^i$.

We now count which blocks can reside in $b$ after the last eviction through $b$ (at time $m_d$).
When a block is accessed and remapped, it gets time stamp $m^*$, which is the current eviction counter.
\begin{itemize}[leftmargin=*]
    \item Blocks with timestamp $m^* \leq m_1$: all $d$ children have been evicted after these blocks were created, so these blocks have been pushed to a child or deeper. They are \emph{not} in $b$.
    \item Blocks with timestamp $m^* > m_d$: these were created after the last eviction through $b$ and have not yet been processed. They are \emph{not} in $b$.
    \item Blocks with timestamp $m_j < m^* \leq m_{j+1}$ for some $1 \leq j \leq d-1$: during the eviction at time $m_d$, \textsc{EvictPath} reads $b$ and writes it back, pushing blocks as deep as possible along the eviction path. Blocks mapped to children $c_{j+1}, \ldots, c_d$ (whose evictions at $m_{j+1}, \ldots, m_d$ occur after the block was created) are pushed to those children. However, blocks mapped to children $c_1, \ldots, c_j$ (whose last evictions at $m_1, \ldots, m_j$ occurred \emph{before} the block was created) cannot be pushed further, because \textsc{EvictPath} only pushes blocks along the current eviction path, and children $c_1, \ldots, c_j$ are not on the eviction paths at times $m_{j+1}, \ldots, m_d$. These blocks remain in $b$.
\end{itemize}

For blocks in the third case with $m_j < m^* \leq m_{j+1}$:
the number of accesses in this window is $A \cdot (m_{j+1} - m_j) = A \cdot d^i$, and each block is mapped to one of $c_1, \ldots, c_j$'s subtrees with probability $j \cdot d^{-(i+1)}$.
Summing over all $j$ from $1$ to $d-1$:
\begin{align*}
    E[\Yb] &= \sum_{j=1}^{d-1} A \cdot d^i \cdot j \cdot d^{-(i+1)} = \frac{A}{d} \sum_{j=1}^{d-1} j \\
             &= \frac{A}{d} \cdot \frac{(d-1)d}{2} = \frac{A(d-1)}{2d} = \frac{A(d-1)}{2}
\end{align*}
\end{proof}

\noindent
With \lemmaref{dary-load} in hand, the remainder of the proof follows RingORAM~\cite{ringoram} exactly, with $a = A(d{-}1)/2$ replacing $a = A/2$ and $(ed)^n$ replacing $4^n$.
For any rooted subtree $T$ with $n = \nT$ nodes, $E[\XT] \leq n \cdot a$.
Since the block indicators $X_i(T)$ are independent (each determined by a fresh random leaf assignment), the Chernoff bound gives:
\[
    \Pr[\XT > \cT + R] \;\leq\;
        \left(\frac{a}{Z}\right)^{\!R}
        \cdot e^{-n\bigl[Z \ln(Z/a) + a - Z\bigr]}.
\]
Let $q = Z \ln(Z/a) + a - Z - 1 - \ln d$.
Substituting into the union bound with the $(ed)^n$ subtree count and summing the geometric series:
\[
    \Pr[\stashst > R]
        \;\leq\; \sum_{n \geq 1}
            \left(\frac{a}{Z}\right)^{\!R}
            \cdot e^{-qn}
        \;=\; \frac{(a/Z)^{R}}{1 - e^{-q}}.
\]
The stash overflow probability decreases exponentially in $R$ whenever $q > 0$, which completes the proof of \theoremref{dary-stash}. \qed

\section{Security Proof}\label{app:security-proof}

\begin{figure}[t]
\centering
\setlength{\fboxsep}{8pt}
\fbox{
\parbox{0.95\linewidth}{
\textbf{Interfaces.}
\begin{itemize}[leftmargin=1.5em]
    \item $\Setup(\idx)$: initialize $\searchsystem$ on index $\idx$ using public parameters $\publicparams$.
    \item Requests are drawn from the public API in \sectionref{system-overview}, namely $\Search(\queryvec, k)$, $\Insert(\point, \pointvec)$, and $\Delete(\point)$.
    During any such operation, whenever the protocol issues an external storage read or write, the request is sent to $\adversary$, which observes the access trace and may return an arbitrary response.
    If any operations outputs abort due to integrity checks, the challenger aborts and the adversary loses.
\end{itemize}

\textbf{Game.}
\begin{enumerate}[leftmargin=1.5em]
    \item Challenger samples a uniformly random bit $b \in \{0,1\}$.
    \item The adversary chooses two equally sized datasets $D_0$ and $D_1$ and one set of public parameters $\publicparams$.
    The challenger constructs the corresponding index $\idx_b$ and runs $\searchsystem.\Setup(\idx_b)$.
    \item The adversary iterates adaptively.
    At step $i$, it chooses a pair of operations $(o_{i,0}, o_{i,1})$, potentially based on all prior observations, subject to the constraint that $o_{i,0}$ and $o_{i,1}$ have the same operation type.
    The challenger then executes the corresponding operation $o_{i,b}$ on $\idx_b$.
    \item Adversary outputs a guess $b'$ and wins if $b' = b$.
\end{enumerate}
}
}
    \caption{Security game for disk-access privacy for a disk-resident ANN search system $\searchsystem$.}\label{fig:disk-security-game}
\end{figure}

\begin{proof}[Proof of \theoremref{disk-privacy}]
We prove indistinguishability via a sequence of hybrids.
Let $\lambda$ denote the security parameter.

\paragraph{Hybrid $H_0$.}
The real security game (\figureref{disk-security-game}) with secret bit $b$.
The challenger runs $\Setup(\idx_b)$ and executes each operation $o_{i,b}$ on $\idx_b$.
The adversary observes physical disk addresses and ciphertexts for each storage access, and can respond arbitrarily to any access request.

\paragraph{Hybrid $H_1$: challenger aborts on any tampering.}
Identical to $H_0$ except that the challenger aborts (and the adversary loses) whenever the adversary returns a ciphertext for slot $(b, k, \oramver)$ that the challenger did not produce for that slot.
$H_0 \approx_c H_1$ because each slot is encrypted with authenticated encryption using associated data $(b \,\|\, k \,\|\, \oramver)$ that uniquely identifies the bucket, physical offset, and version (a monotonic counter incremented on every bucket write).
The challenger validates every response against this associated data and aborts if verification fails.
Thus, the challenger not aborting even on seeing a tampered ciphertext would imply breaking the INT-CTXT property~\cite{authenticated-encryption} of authenticated encryption, which happens with negligible probability.
This means that in $H_1$, the adversary can only win by following the protocol faithfully, returning exactly the ciphertexts the challenger wrote for each slot.
Its only remaining information channel is the pattern of physical disk addresses.

\paragraph{Hybrid $H_2$: replace ciphertexts with encryptions of zeros.}
Identical to $H_1$ except that every ciphertext written to disk encrypts a fixed string $0^B$ (where $B$ is the slot size) instead of the real payload $(\oramaddrs, \oramdata)$.
The challenger now maintains the entire ORAM tree on disk internally, so that it can follow the same access patterns as in $H_1$; only the plaintext inside each ciphertext visible to the adversary changes.
Since each slot is encrypted with a fresh randomness, $H_0 \approx_c H_1$ by IND-CPA security of authenticated encryption~\cite{authenticated-encryption}.

\paragraph{Hybrid $H_3$: switch to the other world.}
Identical to $H_2$ except the challenger runs $\Setup(\idx_{b'})$ and executes $o_{i,b'}$ (where $b' = 1 - b$).
This switch is well-defined because the security game requires $o_{i,0}$ and $o_{i,1}$ to have the same operation type and $|D_0| = |D_1|$, so the challenger can execute the $b'$-world operations using the same public parameters.
Since all ciphertexts encrypt $0^B$ and the adversary must follow the protocol faithfully (from $H_2$), the only observable difference between $H_2$ and $H_3$ is the pattern of physical disk addresses.
We show this pattern is identically distributed by arguing that neither the ANN layer nor the ORAM layer introduces data-dependent disk accesses.

\smallskip\noindent\emph{ANN layer: logical access pattern is data-oblivious.}
The number of ORAM accesses and the access granularity per operation depends only on $\publicparams$ and the operation type:
\begin{itemize}[leftmargin=*]
    \item $\Setup$ writes $2N$ logical blocks ($N$ per ORAM instance), where $N = |D_0| = |D_1|$.
    \item $\Search$ makes exactly $\candidatelistsize$ logical reads of block size $|\adjacencylists| + |\pruninghints|$ from the traversal ORAM and $\prunedlistsize$ logical reads of block size $|\fullprecision|$ from the refinement ORAM (\algorithmref{onyx-anns-algo}).
    $\Insert$ and $\Delete$ also make similar logical accesses based on static public parameters~(\appendixref{dynamic-index}).
\end{itemize}
Since $o_{i,0}$ and $o_{i,1}$ have the same operation type (by the game's constraint), the logical access sequences in $H_2$ and $H_3$ have the same length.

\smallskip\noindent\emph{ORAM layer: physical access pattern is data-oblivious.}
We walk through each subroutine that issues disk I/O and show that the physical addresses depend only on the ORAM's random coins and public state, not on which logical blocks are accessed:
\begin{itemize}[leftmargin=*]
    \item \textsc{ReadPath} (\algorithmref{oram-protocol}): reads one slot per level on path $P(\pmleaf_a, 0), \ldots, P(\pmleaf_a, L)$.
    The leaf $\pmleaf_a$ was drawn uniformly at random from $[d^L]$ when the block was last written, independent of block identity.
    At each level, the physical slot offset is determined by the local pseudorandom permutation $\oramptrs$ stored inside the TEE, invisible to the adversary.
    Whether the target block or a dummy is read, the adversary sees a single slot read at a random offset.

    \item \textsc{EvictPath} (\algorithmref{oram-protocol}): the eviction leaf is computed deterministically from the public access counter $G$ via $\textsc{ReverseDigits}_{d,L}$.
    At each level, the entire bucket is read and then written back---the I/O pattern is fixed and independent of bucket contents.

    \item \textsc{EarlyReshuffle} (\algorithmref{oram-write-helpers}): triggers when $\oramcount = S$ for a bucket on the accessed path.
    The counter $\oramcount$ is public (it equals the number of times the bucket has been accessed since its last reshuffle, which the adversary can track).
    The reshuffle reads and writes the full bucket, a fixed I/O pattern.
\end{itemize}
The two ORAM instances (traversal and refinement) use independent keys, position maps, and random coins, so the argument applies to each instance separately, and their joint access trace is identically distributed in $H_2$ and $H_3$.

\smallskip\noindent
$H_2 \equiv H_3$: the physical access pattern and all ciphertexts are identically distributed.

\paragraph{Hybrid $H_4$: restore real ciphertexts for world $b'$.}
Identical to $H_3$ except ciphertexts now encrypt the real payloads for world $b'$.
By the same IND-CCA2 argument as $H_0 \approx_c H_1$, we have $H_3 \approx_c H_4$.
Note that $H_4$ is exactly the security game with secret bit $b'$.

\paragraph{Conclusion.}
\[
    H_0 \approx_c H_1 \approx_c H_2 \equiv H_3 \approx_c H_4,
\]
where $H_0$ is the game with bit $b$ and $H_4$ is the game with bit $1{-}b$.
Therefore $\Pr[b' = b] \leq \frac{1}{2} + \negl(\lambda)$.
\end{proof}

\end{document}